\documentclass
[11pt,a4paper,english]
{article} 
\usepackage{import}
\usepackage[margin=1in]{geometry}
\usepackage[utf8]{inputenc}
\usepackage{babel}                 
\usepackage{float}   
\usepackage{url}     
\usepackage{sectsty}               
\usepackage{tabularx}              
\usepackage{titling}               
\usepackage{imakeidx}              
\usepackage{xcolor}                
\usepackage{enumitem}              
\usepackage{amsmath}                                    
\usepackage{amssymb}                                    
\usepackage{amsthm}
\usepackage[Gray,squaren,thinqspace,thinspace]{SIunits} 
\usepackage{listings}                                   
\usepackage[small,bf,hang]{caption}        
\usepackage{subcaption}                    
\usepackage{sidecap}                       
\usepackage[colorlinks=true,citecolor=red,linkcolor=blue]{hyperref} 
\usepackage[capitalise,noabbrev]{cleveref}
\usepackage{url}                           
\usepackage{cite}                          
\usepackage{tikz}          
\usepackage{pgfplots}      
\usepackage{pgfplotstable} 
\usepackage{placeins}      
\usepackage{graphicx}      
\usepackage{longtable}     
\usepackage{mathrsfs}
\pgfplotsset{compat=1.18} 
\usepackage{tikz-cd}
\usepackage{braket}
\usetikzlibrary{positioning}
\usetikzlibrary{decorations.pathmorphing}
\usetikzlibrary{knots}
\usetikzlibrary{calc}
\usepackage[export]{adjustbox}

\theoremstyle{plain}
\newtheorem{thm}{Theorem}[section]

\newtheorem{lem}[thm]{Lemma}
\newtheorem{cor}[thm]{Corollary}
\newtheorem*{thm*}{Theorem}

\theoremstyle{definition} 
\newtheorem{ex}{Example}
\newtheorem{defn}[thm]{Definition}
\newtheorem{rmk}{Remark}
\newtheorem{notation}{Notation}

\Crefname{thm}{Theorem}{Theorems}
\Crefname{prop}{Proposition}{Propositions}
\Crefname{lem}{Lemma}{Lemmas}
\Crefname{cor}{Corollary}{Corollaries}
\Crefname{defn}{Definition}{Definitions}
\Crefname{tab}{Table}{Tables}
\Crefname{ex}{Example}{Examples}
\Crefname{chap}{Champter}{Chapters}
\Crefname{app}{Appendix}{Appendices}

\newcommand{\TwoFHilb}{\mathrm{2FHilb}}

\renewcommand{\ket}[1]{| #1 \rangle}

\newcommand{\ra}{\rightarrow}
\newcommand{\End}[1]{\mathrm{End}\br{#1}}
\newcommand{\br}[1]{\left(#1\right)}                               

\newcommand{\ot}{\otimes}

\newcommand{\id}{\mathrm{id}}

\newcommand{\Ftilde}{\widetilde{\mathcal{F}}}
\newcommand{\Stilde}{{\mathcal{S}}}
\newcommand{\Tr}{\mathrm{Tr}}
\newcommand{\Hout}{H_{\mathrm{out}}}
\newcommand{\Hin}{H_{\mathrm{in}}}
\newcommand{\Kout}{K_{\mathrm{out}}}
\newcommand{\Kin}{K_{\mathrm{in}}}

\usepackage{authblk}

\title{Supermaps between channels of any type}
\author[1,2]{Robert Allen\thanks{r.allen@bristol.ac.uk}}
\author[1]{Dominic Verdon}
\affil[1]{School of Mathematics, University of Bristol, UK}
\affil[2]{Heilbronn Institute for Mathematical Research}

\begin{document}

\maketitle

\begin{abstract}
    Supermaps between quantum channels (completely positive trace-preserving (CPTP) maps of matrix algebras) were introduced in [Chiribella \emph{et al.}, EPL 83(3) (2008)]. 
    In this work we generalise to supermaps between channels of any type; by \emph{channels} we mean CPTP maps of finite-dimensional $C^*$-algebras. Channels include POVMs, quantum instruments, classically controlled families of quantum channels, classical channels, quantum multimeters, and more. We show that deterministic supermaps between channels of any type can be realised using simple circuits, recovering the previous realisation theorems of [Chiribella \emph{et al.}, EPL 83(3) (2008)] (for deterministic supermaps between quantum channels) and [Bluhm \emph{et al.} (2024)] (for deterministic supermaps between quantum multimeters) as special cases. To prove this realisation theorem we use the graphical calculus of the 2-category of finite-dimensional 2-Hilbert spaces; the paper includes an accessible and elementary introduction to this graphical calculus, and no prior knowledge of category theory is expected of the reader. 
\end{abstract}

\section{Introduction}

\subsection{Some key concepts}

\paragraph{Channels.} In this work we will consider operations that take in a classical input and a quantum input, and produce a classical output and a quantum output. These operations, called \emph{channels}, include POVMs, quantum instruments, classically controlled families of quantum channels, classical channels, quantum multimeters, and more.

Mathematically, we define a channel to be a completely positive trace-preserving (CPTP) map) $\mathcal{N}: A \to B$ between finite-dimensional $C^*$-algebras $A,B$. To understand how the mathematical definition relates to the operational definition, recall that every finite-dimensional $C^*$-algebra is a multimatrix algebra, i.e. a direct sum of matrix algebras. We therefore have isomorphisms 
\begin{align*}
A \cong \bigoplus_{x \in X} B(H_x)
&&
B \cong \bigoplus_{y \in Y} B(K_y)
\end{align*}
where $X,Y$ are finite sets, $\{H_x~|~x \in X\},\{K_y~|~y \in Y\}$ are indexed families of finite-dimensional Hilbert spaces, and $B(H_x)$ is the algebra of operators (i.e. the matrix algebra) on the Hilbert space $H_x$.  

A \emph{state} (trace-1 positive element) of the algebra $A$ is precisely determined by:
\begin{itemize}
\item A probability distribution $\{p_x\}_{x \in X}$ over the finite set $X$.
\item For each $x \in X$, a quantum state (i.e. a density matrix) $\rho_x \in B(H_x)$.
\end{itemize}
Intuitively, such a state can be thought of as the result of a quantum measurement. The outcome $x \in X$ is measured with some probability, giving the distribution $\{p_x\}$, and the quantum state $\rho_x \in B(H_x)$ is the state resulting from the measurement if the outcome $x$ is obtained. Note that there is here no requirement for the Hilbert spaces $\{H_x\}_{x \in X}$ to be identical.  

A channel $\mathcal{N}: A \to B$ can be interpreted as follows. The input to the channel is a state $\rho \in A$; namely, a probability distribution $\{p_x\}_{x \in X}$ over classical information from the set $x \in X$, which should be interpreted as a \emph{classical control} for the channel, and associated quantum states $\rho_x \in B(H_x)$. If the classical control $x \in X$ is received, a measurement on the associated state $\rho_x$ is performed, with outcome set $Y$. If the outcome $y \in Y$ occurs, the quantum state resulting from the measurement lies in $B(K_y)$. Altogether, the state of the algebra $B$ resulting from this operation is precisely $\mathcal{N}(\rho) \in B$.

Here are some simple examples:
\begin{itemize}
   \item For any finite-dimensional $C^*$-algebra $A$, a channel of type $\mathbb{C} \to A$ is precisely a state of $A$.
   \item A channel of type $B(H) \to \oplus_{x \in X} \mathbb{C}$ is precisely a \emph{positive operator-valued measurement (POVM)} with outcome set $X$. More generally, a channel of type $B(H) \to \oplus_{x \in X} B(K)$ is a \emph{(finite) quantum instrument}, as defined in~\cite{Gud23}.
   \item A channel of type $B(H) \to B(K)$ is precisely a \emph{quantum channel} (CPTP map) between the Hilbert spaces $H,K$, as ordinarily defined. More generally, a channel of type $\oplus_{x \in X} B(H) \to B(K)$ is a classically controlled family of quantum channels between the Hilbert spaces $H,K$.
   \item A channel of type $\oplus_{x \in X} \mathbb{C} \to \oplus_{y \in Y} \mathbb{C}$ is precisely a $|X| \times |Y|$ stochastic matrix. This was described as a `finite discrete memoryless' classical channel in~\cite{Shannon1956}; here we will just call it a \emph{classical channel}.
\end{itemize}

\paragraph{Supermaps.} In the seminal paper~\cite{CDP08} transformations between quantum channels were considered, there called \emph{quantum supermaps}. By `between quantum channels', we mean that that paper focused only on supermaps which map channels of type $B(H_{\mathrm{in}}) \to B(H_{\mathrm{out}})$ to channels of type $B(K_{\mathrm{in}}) \to B(K_{\mathrm{out}})$, where $H_{\mathrm{in}},H_{\mathrm{out}},K_{\mathrm{in}},K_{\mathrm{out}}$ are finite-dimensional Hilbert spaces. Briefly, the results of~\cite{CDP08} can be summarised as follows:
\begin{enumerate}
    \item Based on physical considerations, an abstract mathematical definition of a supermap between quantum channels was presented.
    \item It was shown that any such supermap can be realised operationally using a simple circuit~\cite[Thm. 1]{CDP08}; we will call such a result a `realisation theorem'.
\end{enumerate}
This result has had applications in many areas where quantum channels are transformed or compared, such as cloning \cite{BDPS14}, storage and retrieval of quantum channels \cite{SBZ19}, quantum resource theories \cite{CG19} and error correction \cite{WLWL24}. One key application has been to use the realisation theorem together with semidefinite programming to optimise protocols such as processing reversible quantum channels (including cloning) \cite{BDPS14}, isometry adjointation \cite{YSM24}, computation of quantum correlations \cite{TPKBA24}, and more.

As we have said, the paper~\cite{CDP08} only considered transformations between quantum channels, i.e. from channels of type $B(H_{\mathrm{in}}) \to B(H_{\mathrm{out}})$ to channels of type $B(K_{\mathrm{in}}) \to B(K_{\mathrm{out}})$. The results therefore do not apply to channels of any type. This leaves us unable to answer simple questions about supermaps between other types of channel, such as:
\begin{quote}
What is the most general physical transformation which turns a measurement (POVM) on a given Hilbert space into a state of some other Hilbert space, and how can it be realised?
\end{quote}
The need for generalisation was also highlighted by recent work on transformations between \emph{quantum multimeters}~\cite{Bluhm2024}. These multimeters are measurement apparata which take in a quantum state together with a classical control; the classical control determines the POVM to be performed on the state. In fact, a quantum  multimeter is a channel 
$$
\bigoplus_{x \in X} B(H) \to \mathbb{C}^{\oplus n},
$$
where $H$ is the Hilbert space of the state to be measured, $n$ is the number of possible measurement outcomes and $X$ is the index set for the classical control. Using an \emph{ad hoc} argument the authors of~\cite{Bluhm2024} were able to prove a realisation theorem for supermaps between quantum multimeters.

\subsection{Our results}

\paragraph{A realisation theorem for supermaps between channels of any type.} This raises the question of whether there is a realisation theorem which holds for supermaps between channels of any type,  recovering the known results from~\cite{CDP08,Bluhm2024} as special cases. In this paper we prove such a theorem.

Similarly to both those works, the basis for our definition of supermaps between channels of any type is a generalisation of Choi's theorem. This result (which we will recall more formally in Section~\ref{sec:stinespringchoi}) says that
\begin{quote}
for finite-dimensional $C^*$-algebras $A,B$, there is a bijection (in fact, an isomorphism of convex cones) between completely positive (CP) maps $A \to B$ and positive elements of a certain finite-dimensional $C^*$-algebra which we will call $\underline{\mathrm{Hom}}(A,B)$. This bijection exchanges channels (i.e. trace-preserving CP maps) with positive elements obeying a certain partial trace condition which we will call (TP). 
\end{quote}
A formal, mathematical definition of a supermap between channels of any type can then be given as follows.
\begin{defn}\label{def:detsupermap}
    Let $A,B,C,D$ be f.d. $C^*$-algebras. A \emph{deterministic supermap} from channels of type $A \to B$  to channels of type $C \to D$ is a completely positive map $\underline{\mathrm{Hom}}(A,B) \to \underline{\mathrm{Hom}}(C,D)$ which maps positive elements obeying (TP) to positive elements obeying (TP). 
\end{defn}
\noindent
The justification for this definition is exactly as given in the introduction of~\cite{CDP08}.
A deterministic supermap might also be called a \emph{completely positive trace-preserving-preserving} map. 

The main result of this paper is the following realisation theorem for these deterministic supermaps. 
\begin{thm}\label{thm:mainintro}
Consider the finite-dimensional $C^*$-algebras
\begin{align*}
&A := \bigoplus_{i \in I} B(H_{\mathrm{in},i}) &&
&B := \bigoplus_{j \in J} B(H_{\mathrm{out},j}) \\
&C := \bigoplus_{k \in K} B(K_{\mathrm{in},k})
&&&D := \bigoplus_{l \in L} B(K_{out,l}),
\end{align*}
where $\{H_{\mathrm{in},i}\}_{i \in I}$, $\{H_{\mathrm{out},j}\}_{j \in J}$, $\{K_{\mathrm{in},k}\}_{k \in K}$, $\{K_{\mathrm{out},l}\}_{l \in L}$ are finite sets of f.d. Hilbert spaces. 

Let $\mathcal{S}$ be any deterministic supermap which transforms channels of type $A \to B$ to channels of type $C \to D$. Then there exists
\begin{itemize}
    \item a Hilbert space $P$  of dimension at most $$\max_{i \in I, k \in K}\left(\dim(\Hin{}_{,i})\dim(\Kin{}_{,k})\right),$$
    \item a channel 
$$\mathcal{E}: C \to A \otimes B(P),$$
\item and a channel 
$$
\mathcal{G}: \bigoplus_{\substack{i \in I\\ k \in K}} B \otimes B(P) \to D,
$$
\end{itemize}
such that the supermap may be realised operationally as follows. Let $\mathcal{F}: A \to B$ be any channel. Then $\mathcal{S}(\mathcal{F}): C \to D$ is described by the following circuit:
\begin{align}\label{eq:realisationcircuit}
\includegraphics[valign=c,scale=1.3]{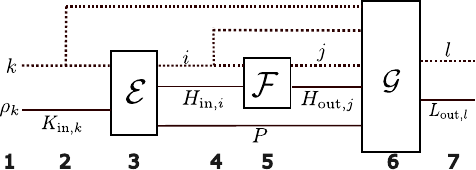}
\end{align}
\end{thm}
\noindent
The circuit should be read from left to right; we describe it in words now. 
\begin{enumerate}
    \item The input to the channel $\mathcal{S}(\mathcal{F})$ is a state of type $C$: a probability distribution over classical data $k \in K$ and associated quantum states $\rho_k \in B(K_{\mathrm{in},k})$.
    \item The classical data $k \in K$ is copied.
    \item One copy of the classical data $k \in K$, together with the associated state $\rho_k$, are given as input to the channel $\mathcal{E}$.
    \item The measurement outcome $i \in I$ resulting from the channel $\mathcal{E}$ is copied. 
    \item One copy of $i \in I$, together with the associated state of $B(H_{\mathrm{in},i})$ output by the channel $\mathcal{E}$, is given as input to the channel $\mathcal{F}$.
    \item All the classical measurement outcomes $i \in I, j\in J, k\in K$, together with the state of $B(H_{\mathrm{out},j}) \otimes B(P)$ output by the previous channels, are given as input to the channel $\mathcal{G}$.
    \item The output is a state of type $D$: a probability distribution over classical measurement outcomes $l \in L$ and associated quantum states $\sigma_l \in B(K_{\mathrm{out},l})$.
\end{enumerate}
To illustrate the theorem we will now give some examples. We will begin by recovering the results of~\cite{CDP08,Bluhm2024}.

\begin{ex}[Quantum channels to quantum channels]
We begin by recovering the realisation theorem of~\cite{CDP08}. Let 
\begin{align*}
A:= B(\Hin) && B:=B(\Hout) && C:=B(\Kin) && D:= B(\Kout)
\end{align*}
In this case all the classical information is trivial and the circuit reduces to
$$
\includegraphics[valign=c,scale=1.3]{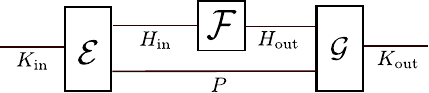}
$$
which is precisely the circuit in~\cite[Fig. 1]{CDP08}, where the ancilla output by $\mathcal{G}$ is discarded.
\end{ex}

\begin{ex}[Multimeters to multimeters]
We now recover the realisation theorem of~\cite{Bluhm2024}. Let
\begin{align*}
A:= \bigoplus_{i \in I} B(H) 
&&
B:= \bigoplus_{j \in J} \mathbb{C}
&&
C:= \bigoplus_{k \in K} B(K)
&&
D:= \bigoplus_{l \in L} \mathbb{C}
\end{align*}
In this case the circuit reduces to
$$
\includegraphics[valign=c,scale=1.3]{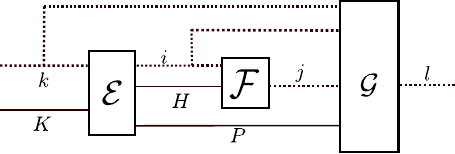}
$$
which is precisely the circuit in~\cite[Fig. 5]{Bluhm2024}. 
\end{ex}

\begin{ex}[POVMs to quantum states] We will now answer the question posed earlier in the introduction: what is the most general supermap which maps POVMs to quantum states? Let 
\begin{align*}
A:= B(H) 
&&
B:=\bigoplus_{j \in J} \mathbb{C}
&&
C:=\mathbb{C}
&&
D:=B(K)
\end{align*}
The circuit then reduces to
$$
\includegraphics[valign=c,scale=1.3]{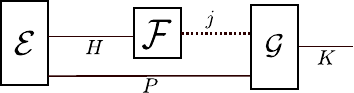}
$$
and we see that the most general way to transform POVMs into quantum states is to create a bipartite quantum state, perform the POVM on one of the parts, and use the outcome of the POVM to control a quantum channel on the remainder. 
\end{ex}

\begin{ex}[Quantum states to POVMs]
Going the other way, one might ask what is the most general way to transform a quantum state into a POVM. Let:
\begin{align*}
A:= \mathbb{C}
&&
B:= B(H)
&&
C:=B(K)
&&
D:=\bigoplus_{l \in L} \mathbb{C}
\end{align*}
The circuit then reduces to
$$
\includegraphics[valign=c,scale=1.3]{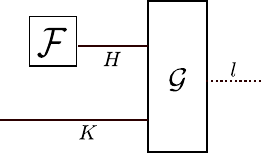}
$$
and we see that the most general way to transform a quantum state into a POVM is to perform a bipartite POVM on the state together with the system to be measured.
\end{ex}
\begin{ex}[Classical channels to quantum channels]
What is the most general way to transform a classical channel into a quantum channel?
Let:
\begin{align*}
A:=\bigoplus_{i \in I} \mathbb{C}
&&
B:=\bigoplus_{j \in J} \mathbb{C}
&&
C:=B(H)
&&
D:=B(K)
\end{align*}
The circuit then reduces to
$$
\includegraphics[valign=c,scale=1.3]{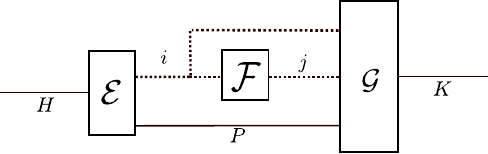}
$$
and we see that the most general way to transform a classical channel into a quantum channel is to perform a quantum instrument whose quantum output lies in some Hilbert space $P$ of dimension at most $\dim(H)$, copy the measurement outcome, use the classical channel to postprocess one copy, and then use the measurement outcome and its postprocessed counterpart as controls on some quantum channel $B(P) \to B(K)$. 
\end{ex}

\begin{ex}[Quantum channels to classical channels]
What is the most general way to transform a quantum channel into a classical channel? Let
\begin{align*}
A:= B(H)
&&
B:=B(K)
&&
C:=\bigoplus_{k \in K} \mathbb{C}
&&
D:=\bigoplus_{l \in L} \mathbb{C}
\end{align*}
The circuit then reduces to
$$
\includegraphics[valign=c,scale=1.3]{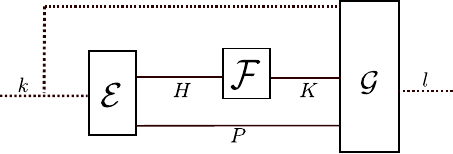}
$$
and we see that the most general way to transform a quantum channel into a classical channel is to copy the classical input, use one copy to control a bipartite quantum state initialisation, put one part of the initialised bipartite state through the channel, and then perform a POVM controlled by the remaining copy of the classical input.
\end{ex}
\noindent
We leave the entertainment of discovering further examples to the reader.
\paragraph{Future applications.} As we discussed already, supermaps between quantum channels have found applications in many areas of mathematics and physics. Possessing a realisation theorem for deterministic supermaps between channels of any type allows us to extend these applications. In particular, any protocol transforming, comparing, approximating, programming or simulating a channel of any type can be decomposed and optimised analogously to the case of quantum channels using a semidefinite programming approach. We leave these applications for future work. 

\subsection{The structure of the the paper}
The remainder of this paper will be devoted to the proof of Theorem~\ref{thm:mainintro}. 

Our strategy for proving this is quite straightforward. We take the original proof for supermaps between quantum channels given in~\cite{CDP08}; we formulate it using the diagrammatic calculus of the 2-category $\TwoFHilb$ of finite-dimensional 2-Hilbert spaces; and thereby generalise it to supermaps between channels of any type. We propose this as a broadly applicable method of turning results about quantum channels into results about channels of any type; for another example of a similar approach see~\cite{Verdon2024}, which used $\TwoFHilb$ to generalise quantum teleportation and dense coding protocols.

In Section~\ref{sec:background} we give a short and self-contained introduction to the diagrammatic calculus of $\TwoFHilb$, and a review of the diagrammatic representation of the generalised Stinespring's and Choi's theorems.

In Section~\ref{sec:proof} we present the proof of Theorem~\ref{thm:mainintro}.

\subsection{Acknowledgements}
RA was supported by the Additional Funding Programme for Mathematical Sciences, delivered by EPSRC (EP/V521917/1) and the Heilbronn Institute for Mathematical Research. 
DV was supported by the European Research Council (ERC) under the European Union’s Horizon 2020 research and innovation program
(Grant Agreement No. 817581). 

\section{Background}\label{sec:background}

\subsection{An introduction to the diagrammatic calculus}\label{sec:diagcalc}

The diagrammatic calculus we present in this paper is, formally, the diagrammatic calculus of the 2-category $\TwoFHilb$ of finite-dimensional (f.d.) 2-Hilbert spaces~\cite{Baez96}. This graphical calculus was first proposed for application to quantum information theory in~\cite{Vicary2012,Vicary2012a}; for further work in this direction see e.g.~\cite{Reutter2019,Verdon2024,Claeys2024}.
A full and accessible introduction to $\TwoFHilb$ and its graphical calculus, and to category theory more generally, can be found in the book~\cite{Heunen2019}. Accessible introductions are also given in~\cite{Reutter2019,Verdon2024}; these provide further details and may help the reader if our exposition is at any point unclear.

In order to make our exposition accessible to a broad audience, we will avoid discussion of category theory altogether and simply introduce the diagrammatic calculus of $\TwoFHilb$ as a way to depict indexed families of linear maps, following the approach taken in~\cite{Reutter2019}. We would not recommend working with the graphical calculus without understanding its categorical underpinnings (for this see e.g.~\cite[\S{}8]{Heunen2019}), but our elementary overview should be sufficient for the reader to understand our proof.  

\subsubsection{Interpreting 2-morphism diagrams}

The essential difference between the diagrammatic calculus we will use in this paper and the well-known calculus of \emph{tensor diagrams} (a.k.a. \emph{wire diagrams} or \emph{string diagrams}) (described in e.g.~\cite[Chap. 4]{Coecke2017}) is that the diagrams here include planar regions which are connected by the wires. We call such diagrams \emph{2-morphism diagrams}. Here is an example:
\begin{align}\label{eq:2morphex}\includegraphics[valign=c]{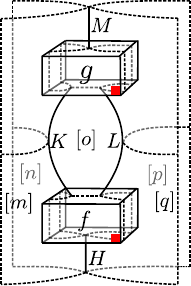}
\end{align}
While tensor diagrams represent linear maps, 2-morphism diagrams represent \emph{indexed families} of linear maps. We will shortly explain how to interpret these diagrams, but will first discuss their basic constituents: planar regions, wires and boxes.

\paragraph{Regions and objects.} The planar regions in a 2-morphism diagram are labelled by \emph{index sets}. We write $[n]$ for the set $\{1,\dots,n\}$, where $n \in \mathbb{N}$; all our index sets are sets of this type. If a region is labelled by $[1]$ it will be invisible in the graphical calculus.

We can compose planar regions in only one way: by layering them in front of each other. We will use the notation $\boxtimes$ for this composition, where the plane at the back is on the left of the product symbol and the plane at the front is on the right of the product symbol. For instance, suppose we have a planar region labelled by $[m]$ and a planar region labelled by $[n]$; then we can draw the planar region $[m] \boxtimes [n]$, depicted as follows:
\begin{align*}
\includegraphics[valign=c]{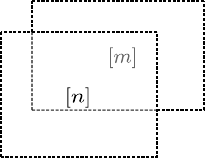}
\end{align*}
In the language of category theory, any planar region, or more generally any composition of planar regions, is called an \emph{object}.

\paragraph{1-morphisms.} 1-morphisms are maps from one object to another. In the graphical calculus they are depicted as wires which travel up the page and join an object on the left of the wire --- the \emph{source} --- to an object on the right of the wire --- the \emph{target}. Every 1-morphism represents an indexed family of f.d. Hilbert spaces; these families are indexed by the index sets labelling the planar regions in the source and target. We say that a 1-morphism is of \emph{type} $O_1 \to O_2$, where $O_1$ is the source and $O_2$ is the target. For instance, consider the following 1-morphism:
\begin{align*}
\includegraphics[valign=c]{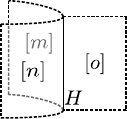}
\end{align*}
This 1-morphism $H$ is of type $[m] \boxtimes [n] \to [o]$. The planar regions in the source and target have index sets $[m]$, $[n]$ and $[o]$. The 1-morphism therefore represents an indexed family of Hilbert spaces $\{H_{mno}\}_{m,n,o \in [m] \times [n] \times [o]}$. 

Note that 1-morphism wires are drawn with a solid black line, whereas the other lines in the diagram, which only serve to indicate the position of the planar regions, are dashed. This convention will be used throughout this work.

We can compose 1-morphisms in two ways. Firstly, we can layer them in front of each other, just like planar regions. We again use the $\boxtimes$ notation for this composition. For instance, suppose we have a 1-morphism $H: [m] \to [n]$ and a 1-morphism $K: [p] \to [q]$; then we can form the composition $H \boxtimes K: [m] \boxtimes [p] \to [n] \boxtimes [q]$, depicted as follows:
\begin{align*}
\includegraphics[valign=c]{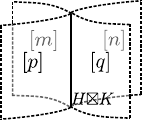}
~~:=~~
\includegraphics[valign=c]{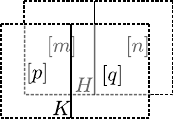}
\end{align*}
The composition $H \boxtimes K$ again represents a family of Hilbert spaces $$\{(H \boxtimes K)_{mnpq}\}_{m,n,p,q \in [m] \times [n] \times [p] \times [q]}$$ indexed by the surrounding planar regions. These Hilbert spaces are defined using the tensor product:
$$
(H \boxtimes K)_{mnpq} := H_{mn} \otimes K_{pq}
$$
The other way we can compose 1-morphisms is as follows. Suppose we have two 1-morphisms $H: O_1 \to O_2$ and $K: O_2 \to O_3$. Since the object on the right of $H$ is the same as the object on the left of $K$, we can compose the wires `horizontally', across the object $O_2$. We will use the notation $\otimes$ for this composition. Suppose for example that $O_1:= [m]$, $O_2:=[n]$ and $O_3:=[o]$; then $H \otimes K: [m] \to [o]$ is depicted as follows:
\begin{align*}
\includegraphics[valign=c]{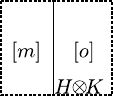}
~~:=~~
\includegraphics[valign=c]{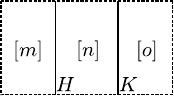}
\end{align*}
Again, the composition $H \otimes K$ represents a family of Hilbert spaces 
$$
\{(H \otimes K)_{mo}\}_{m,o \in [m] \times [o]}
$$
indexed by the planar regions of the source and target. These Hilbert spaces are defined using the tensor product and direct sum:
$$
(H \otimes K)_{mo} := \bigoplus_{n \in [n]} H_{mn} \otimes K_{no}
$$

\paragraph{Direct sum of 1-morphisms.} We have mentioned that there are two ways to compose 1-morphisms in a general 2-category. Because of the linear structure of $\TwoFHilb$, it possesses a third type of composition which we will discuss now.

Suppose that $H$ and $K$ are 1-morphisms of the same type; for example, suppose that $H,K: [m] \to [n]$. Then we can take the \emph{direct sum} $H \oplus K: [m] \to [n]$. Again, this 1-morphism represents an indexed family of Hilbert spaces 
$$
\{(H \oplus K)_{mn}\}_{m,n \in [m] \times [n]}.
$$
These Hilbert spaces are defined using the componentwise direct sum in the obvious way:
\begin{align}\label{eq:directsum1morphisms}
(H \oplus K)_{mn}:= H_{mn} \oplus K_{mn}
\end{align}

\paragraph{2-morphisms.} 2-morphisms are maps from one 1-morphism to another. The two 1-morphisms must have the same type. In the graphical calculus 2-morphisms are depicted as boxes. The \emph{source} 1-morphism enters at the bottom of the box, and the \emph{target} 1-morphism leaves at the top of the box. We say that the box is of \emph{type} $H \to K$, where $H$ is the source 1-morphism and $K$ is the target 1-morphism. 2-morphisms represent indexed families of linear maps, indexed by the planar regions which are the source and target of the source and target 1-morphisms. 

For example, suppose we have two 1-morphisms $H, K: [m] \boxtimes [n] \to [o]$. Since these 1-morphisms are of the same type, we can have a 2-morphism $f: H \to K$ between them. This 2-morphism is depicted as a box:
$$
\includegraphics[valign=c]{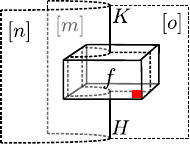}
$$
The reader will see that the 1-morphism $H$ enters in at the bottom of the box, and the 1-morphism $K$ leaves at the top. Notice that the box has a red square in the bottom-right corner of the front face; this notation allows us to easily depict the dagger, transpose and conjugate of a 2-morphism, as we will explain shortly. Since the 1-morphisms have source $[m] \boxtimes [n]$ and target $[o]$, the 2-morphism $f$ represents an indexed family of linear maps:
$$
f:=\{f_{mno}: H_{mno} \to K_{mno}\}_{m,n,o \in [m] \times [n] \times [o]}
$$
We can compose 2-morphisms in three ways. Firstly, we can layer them in front of each other. We use the notation $\boxtimes$ for this composition. Suppose for example that we have wires $H,K: [m] \to [n]$ and $L,M: [o] \to [p]$, and boxes $f: H \to K$ and $g: L \to M$. Then we obtain a box $f \boxtimes g: H \boxtimes L \to K \boxtimes M$:
$$
\includegraphics[valign=c]{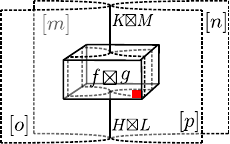}
~~:=~~
\includegraphics[valign=c]{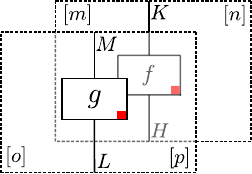}
$$
Here we have flattened the boxes on the right hand side; this is because they each have only one plane touching at the left and right, so there is no need to make them three-dimensional. We will do this frequently later on. As discussed above, this composition of boxes represents an indexed family of linear maps 
$$
\{(f \boxtimes g)_{mnop}: H_{mn} \otimes L_{op} \to K_{mn} \otimes M_{op}\}_{m,n,o,p \in [m] \times [n] \times [o] \times [p]}.
$$ 
These linear maps are defined using the tensor product:
$$
(f \boxtimes g)_{mnop}:= 
f_{mn} \otimes g_{op}
$$
Secondly, we can compose 2-morphisms `horizontally', provided the types of the 1-morphisms match (what we mean by this will be clarified by the example). We use the notation $\otimes$ for this composition. Suppose for example that we have 1-morphisms $H,K: [m] \to [n]$ and $L,M: [n] \to [o]$, and boxes $f: H \to K$ and $g: L \to M$. Since the target of the 1-morphisms $H,K$ is the same as the source of the 1-morphisms $L,M$, we can compose the boxes across that object to obtain a 2-morphism $f \otimes g: H \otimes L \to K \otimes M$:
$$
\includegraphics[valign=c]{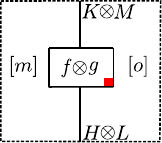}
~~:=~~
\includegraphics[valign=c]{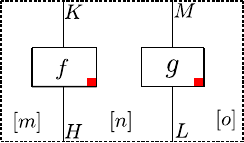}
$$
This composition represents an indexed family of linear maps
$$
\{(f \otimes g)_{mo}: \bigoplus_{n \in [n]} H_{mn} \otimes L_{no} \to \bigoplus_{n \in [n]} K_{mn} \otimes M_{no}\}_{m,o \in [m] \times [o]}.
$$
These linear maps are defined using the tensor product and direct sum:
$$
(f \otimes g)_{mo} := \bigoplus_{n \in [n]} f_{mn} \otimes g_{no}
$$
Third, we can compose 2-morphisms `vertically', provided that the source of one 2-morphism is identical with the target of the other. We use the notation $\circ$ for this composition. Suppose for example that we have 1-morphisms $H,K,L: [m] \to [n]$ and 2-morphisms $f: H \to K$ and $g: K \to L$. Then we can form the vertical composition $g \circ f: H \to L$:
$$
\includegraphics[valign=c]{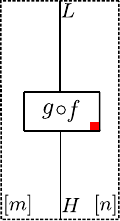}
~~:=~~
\includegraphics[valign=c]{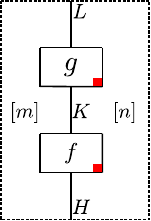}
$$
This composition represents a family of linear maps 
$$
\{(g \circ f)_{mn}: H_{mn} \to L_{mn}\}_{m,n \in [m] \times [n]}
$$
indexed by the adjoining planar regions. These are defined using ordinary composition of linear maps:
$$
(g \circ f)_{mn} := g_{mn} \circ f_{mn}
$$
We finish by mentioning \emph{identity 2-morphisms}. Let $O_1,O_2$ be objects and $H: O_1 \to O_2$ be a 1-morphism. We write $\id_H: H \to H$ for the identity 2-morphism on $H$. This 2-morphism represents a family of linear maps indexed by the planar regions in $O_1,O_2$, all of which are the identity map. Identity 2-morphisms are invisible in the graphical calculus. 

\paragraph{Direct sum of 2-morphisms.} The linear structure of $\TwoFHilb$ gives rise to another composition of 2-morphisms. Let $H,K,L$ be 1-morphisms of the same type, and let $f: H \to K$ and $g: H \to L$ be 2-morphisms. Then there is a 2-morphism 
$$f \oplus g: H \to K \oplus L$$
given by the componentwise direct sum of linear maps. For example, suppose that $H,K,L: [m] \to [n]$; then $f \oplus g$ represents a family of linear maps
$$
\{(f \oplus g)_{mn}\}_{m,n \in [m] \times [n]}
$$
defined by 
\begin{align}\label{eq:directsum2morphisms}
(f \oplus g)_{mn} := \begin{pmatrix} f_{mn} \\ g_{mn} \end{pmatrix}: H_{mn} \to K_{mn} \oplus L_{mn},
\end{align}
where we have used matrix notation to express the linear map $H_{mn} \to K_{mn} \oplus L_{mn}$ obtained from the linear maps $f_{mn}: H_{mn} \to K_{mn}$ and $g_{mn}: H_{mn} \to L_{mn}$ by the universal property of the direct sum. 
\paragraph{Interpretation.} We can now explain how to interpret a general 2-morphism diagram. The simple rule to follow is that a 2-morphism diagram represents a family of linear maps indexed by the open planar regions in the diagram, while the index sets for closed planar regions are summed over. This is probably best seen by example. Let us revisit the example of a 2-morphism given above:
$$
\includegraphics[valign=c]{graphics/graphcalc/2morphex.pdf}
$$
Based on what we have said above we can analyse this diagram. 
There are five different planar regions in the diagram, with index sets $[m],[n],[o],[p],[q]$. There are four wires of the following types:
\begin{align*}
H: [n] \boxtimes [m] \to [p] \boxtimes [q]
&&
K: [n] \boxtimes [m] \to [o]
\\
L: [o] \to [p] \boxtimes [q]
&&
M: [n] \boxtimes [m] \to [p] \boxtimes [q]
\end{align*}
These four wires represent indexed families of Hilbert spaces:\begin{align*}
\{H_{mnpq}\}_{m,n,p,q \in [m] \times [n] \times [p] \times [q]}
&&
\{K_{mno}\}_{m,n,o \in [m] \times [n] \times [o]} 
\\
\{L_{opq}\}_{o,p,q \in [o] \times [p] \times [q]} 
&&
\{H_{mnpq}\}_{m,n,p,q \in [m] \times [n] \times [p] \times [q]}
\end{align*}
There are two boxes of the following types:
\begin{align*}f: H \to K \otimes L
&&
g: K \otimes L \to M
\end{align*}
These two boxes represent indexed families of linear maps:
\begin{align*}
\{f_{mnpq} : H_{mnpq} \to \bigoplus_{o \in [o]} K_{mno} \otimes L_{opq}\}_{m,n,p,q \in [m] \times [n] \times [p] \times [q]}
\\
\{g_{mnpq} : \bigoplus_{o \in [o]} K_{mno} \otimes L_{opq} \to M_{mnpq}\}_{m,n,p,q \in [m] \times [n] \times [p] \times [q]}
\end{align*}
By the universal property of the direct sum, the linear map $f_{mnpq}$ is equivalent to a set of linear maps 
$$
\{f_{mnopq}: H_{mnpq} \to K_{mno} \otimes L_{opq}\}_{o \in [o]}.
$$
Likewise, $g_{mnpq}$ is equivalent to a set of linear maps 
$$
\{g_{mnopq}: K_{mno} \otimes L_{opq} \to M_{mnpq}\}_{o \in [o]}
$$
The diagram is a vertical composition representing a 2-morphism of type $H \to M$. There are four open planar regions and one closed planar region. Summing over the index set for the closed planar region, we see that the 2-morphism diagram represents the following indexed family of linear maps:
$$
\{\sum_{o \in [o]}g_{mnopq} \circ f_{mnopq}: H_{mnpq} \to M_{mnpq}\}_{m,n,p,q \in [m] \times [n]\times [p] \times[q]}
$$

\subsubsection{Duality for 1-morphisms}
We will now discuss duality for 1-morphisms. 

Let $O_1,O_2$ be objects and let $H: O_1 \to O_2$ be a 1-morphism, which as we have seen represents an indexed family of Hilbert spaces.

We will now define a 1-morphism $H^*: O_2 \to O_1$ representing the same indexed family of Hilbert spaces, except that now the object $O_2$ is on the left, and the object $O_1$ is on the right. This 1-morphism $H^*$ will be called the \emph{dual} of $H$. 

We will first show how the dual is defined, then show how it behaves with respect to composition. 

\paragraph{Duals of wires.}
For example, then, let $H: [m] \boxtimes [n] \to [o]$ be a 1-morphism. As we have seen, this 1-morphism represents an indexed family of Hilbert spaces $\{H_{mno}\}_{m,n,o \in [m] \times [n] \times [o]}$.

The dual is a 1-morphism $H^*: [o] \to [m] \boxtimes [n]$ which again represents the indexed family of Hilbert spaces $\{H_{mno}\}_{m,n,o \in [m] \times [n] \times [o]}$. Note that the family of Hilbert spaces is the same, but the source and target have been switched.

In the diagrammatic calculus, to distinguish the wire for $H$ from the wire for $H^*$, we will label the former with an upwards pointing arrow and the latter with a downwards pointing arrow. Here are the depictions of the 1-morphisms $H$ (on the left) and $H^*$ (on the right):
\begin{align*}
\includegraphics[valign=c]{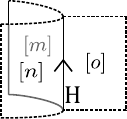}
&&
\includegraphics[valign=c]{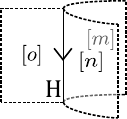}
\end{align*}

\paragraph{Duals under composition.}
The above definition of the dual holds for any 1-morphism. However, for composite 1-morphisms there is an alternative and categorically equivalent way to define the dual. This is specified by the following equations:
\begin{align*}
(H \boxtimes K)^* \cong H^* \boxtimes K^*
&&
(H \otimes K)^* \cong K^* \otimes H^*.
\end{align*}
In category theory the $\cong$ symbol indicates natural isomorphism, but we will not worry about what this means here; all that it is necessary to know is that one may equivalently define the dual according to the prescription on the RHS of the above equations. 

We will give two examples of how this works in diagrams. For the first example, let $H: [m] \to [n]$ and $K: [o] \to [p]$ be 1-morphisms. The left diagram below shows $H \boxtimes K: [m] \boxtimes [o] \to [n] \boxtimes [p]$, and the right diagram shows $H^* \boxtimes K^*: [n] \boxtimes [p] \to [m] \boxtimes [o]$:
\begin{align*}
\includegraphics[valign=c]{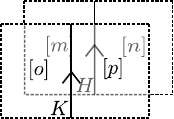}
&&
\includegraphics[valign=c]{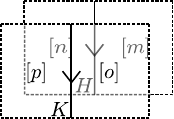}
\end{align*}
For the second example, let $H: [m] \to [n]$ and $K: [n] \to [o]$ be 1-morphisms. The left diagram below shows $H \otimes K: [m] \to [o]$ and the right diagram shows $K^* \otimes H^*: [o] \to [m]$:
\begin{align*}
\includegraphics[valign=c]{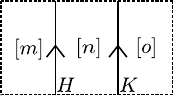}
&&
\includegraphics[valign=c]{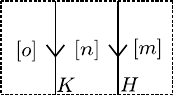}
\end{align*}

\paragraph{Cups and caps.} In category theory, the `duality' between a 1-morphism $H$ and its dual $H^*$ is witnessed by a pair of special 2-morphisms called a `cup and cap'. We will now define this cup and cap and explain the equations which characterise them. 

Let $H: O_1 \to O_2$ be a 1-morphism. The \emph{cup} is a 2-morphism of type $\eta_H: - \to H^* \otimes H$ and the \emph{cap} is a 2-morphism of type $\epsilon_H: H \otimes H^* \to -$. Here we use the $-$ to indicate that the input or output of the box are empty; formally, this is to say that the input or output is an \emph{identity 1-morphism}, but that is something we will gloss over in this basic introduction. We will draw the cup and cap boxes in a topologically suggestive way, as bends in the wire. For example, let $H: [m] \boxtimes [n] \to [o]$ be a wire; then the cup and cap are depicted as follows:
\begin{align}\label{eq:cupcap3d}
\includegraphics[valign=c]{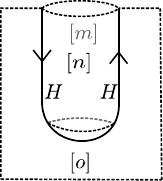}
&&
\includegraphics[valign=c]{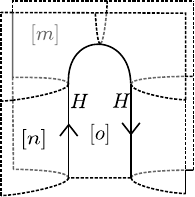}
\end{align}
Like all 2-morphisms, the cup and cap represent families of linear maps. We will show how to determine these linear maps for the example 1-morphism $H: [m] \otimes [n] \to [o]$; the definition for a general 1-morphism is similar.

By the type of the cup $\eta_H$, we see that it represents an indexed family of linear maps $\{\eta_{mno}: \mathbb{C} \to H_{mno} \otimes H_{mno}\}_{m,n,o \in [m] \times [n] \times [o]}$. (We have here used the convention that a box whose input is empty represents linear maps from $\mathbb{C}$.) Such a linear map is determined by a single vector, namely the image of the unit $1$. We define:
$$
\eta_{mno}: 1 \mapsto \sum_{i \in I} \ket{i} \otimes \ket{i}
$$
where $\{\ket{i}\}_{i \in I}$ is some chosen orthonormal basis for $H_{mno}$. This defines the cup. The cup for any 1-morphism is defined similarly.

By the type of the cap $\epsilon_H$, we see that it represents an indexed family of linear maps $\{\epsilon_{mno}: H_{mno} \otimes H_{mno} \to \mathbb{C}\}$. (We have here used  the convention that a box whose output is empty represents linear maps to $\mathbb{C}$.) Using the same orthonormal basis, we define these maps as follows on basis elements: 
$$
\epsilon_{mno}: \ket{i} \otimes \ket{j} \mapsto \braket{i|j}
$$
The cap for an arbitrary 1-morphism is defined similarly. 

\paragraph{Snake equations.} The point of the cup and cap is that they obey the following `snake' or `zigzag' equations: \begin{align}\label{eq:snakes}
\includegraphics[valign=c]{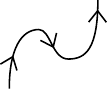}
~~=~~
\includegraphics[valign=c]{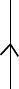}
&&
\includegraphics[valign=c]{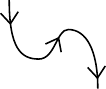}
~~=~~
\includegraphics[valign=c]{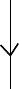}
\end{align}
In these equations we drew only the 1-morphism wire and ignored the planar regions. To see what the snake equations look like with the planar regions included, here they are for the 1-morphism $H: [m] \times [n] \to [o]$, for example:
\begin{align*}
\includegraphics[valign=c,scale=0.9]{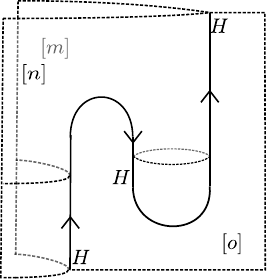}
&=
\includegraphics[valign=c,scale=0.9]{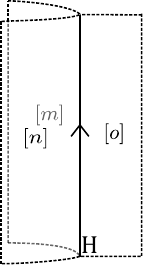}
\end{align*}
\begin{align*}
\includegraphics[valign=c,scale=0.9]{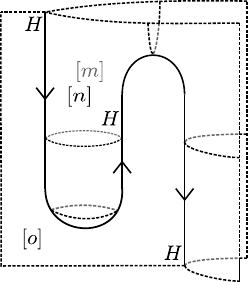}
&=
\includegraphics[valign=c,scale=0.9]{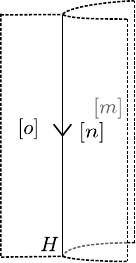}
\end{align*}

\paragraph{Cup and cap for compositions of 1-morphisms.} Above we explained an alternative way to define the dual for compositions of 1-morphisms. Here we will explain how the cup and cap are defined for those duals. 

For the composition $\boxtimes$, we have the following rule:
\begin{align*}
\epsilon_{H \boxtimes K}:= \epsilon_H \boxtimes \epsilon_K
&&
\eta_{H \boxtimes K}:= \eta_H \boxtimes \eta_K
\end{align*}
For example, let $H: [m] \to [n]$ and $K: [o] \to [p]$ be 1-morphisms; then the cup and cap $\eta_{H \boxtimes K}$ and $\epsilon_{H \boxtimes K}$ are as follows:
\begin{align*}
    \includegraphics[valign=c]{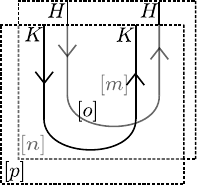}
&&
\includegraphics[valign=c]{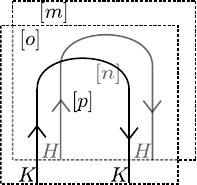}
\end{align*}
For the horizontal composition $\otimes$, the cups and caps are `nested' according to the following rule:
\begin{align*}
\epsilon_{H \otimes K} := \epsilon_H \circ (\id_H \otimes \epsilon_K \otimes \id_{H^*})
&&
\eta_{H \otimes K} := (\id_{K^*} \otimes \eta_H \otimes \id_{K}) \circ \eta_{K}
\end{align*}
For example, let $H: [m] \to [n]$ and $K:[n] \to [o]$ be 1-morphisms, then the cup and cap $\eta_{H \otimes K}$ and $\epsilon_{H \otimes K}$ are as follows:
\begin{align*}
\includegraphics[valign=c]{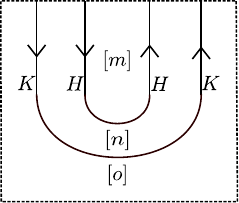}
&&
\includegraphics[valign=c]{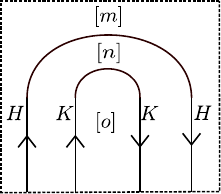}
\end{align*}

\subsubsection{Dagger, transpose and conjugate.}
\paragraph{Dagger.} Let $H,K$ be 1-morphisms of the same type, and let $f: H \to K$ be a 2-morphism. The \emph{dagger} of $f$ is a 2-morphism 
$$
f^{\dagger}: K \to H
$$
We will define the dagger of a 2-morphism by example, since the example we will present generalises straightforwardly to any 2-morphism. For example, then, let $H: [m] \to [o]$, $K:[m] \to [n]$ and $L: [n] \to [o]$ be 1-morphisms, and consider a box $f: H \to K \otimes L$:
$$
\includegraphics[valign=c]{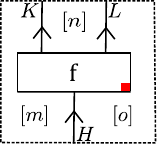}
$$
As we have seen, this box represents an indexed family of linear maps
$$\{f_{mo}: H_{mno} \to \bigoplus_{n \in [n]} K_{mn} \otimes L_{no}\}_{m,o \in [m] \times [o]}.
$$
The dagger of $f$ is a box $f^{\dagger}: K \otimes L \to H$, labelled with $f$ but with the red square on the top right rather than the bottom right corner, as if it had been reflected in a horizontal axis while preserving the direction of the arrows on the wires:
$$
\includegraphics[valign=c]{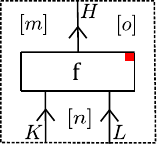}
$$
This box $f^{\dagger}$ represents the indexed family of linear maps
$$
\{(f_{mo})^{\dagger}: \bigoplus_{n \in [n]} K_{mn} \otimes L_{no} \to H_{mo}\}_{m,o \in [m] \times [o]}
$$
where the dagger symbol for a linear map indicates the Hermitian adjoint. 

To take the dagger of a composition of 2-morphisms, one simply reflects the whole diagram in a horizontal axis while preserving the direction of the arrows on the wires. For example, the dagger of~\eqref{eq:2morphex} is given by the following 2-morphism diagram:
$$
\includegraphics[valign=c]{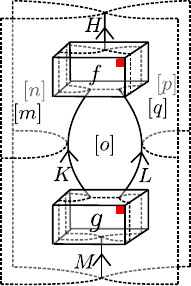}
$$
Using what has already been said about interpreting 2-morphism diagrams, the reader should easily be able to determine that this 2-morphism represents the following family of linear maps:
$$
\{\sum_{o \in [o]} (f_{mnopq})^{\dagger}\circ (g_{mnopq})^{\dagger}: M_{mnpq} \to H_{mnpq}\}_{m,n,p,q \in [m] \times [n] \times [p] \times [q]}
$$
Here $\{f_{mnopq}\}_{o \in [o]}$ and $\{g_{mnopq}\}_{o \in [o]}$ have been defined using the universal property of the direct sum, as in the above paragraph on interpretation of 2-morphism diagrams. 

Before moving on, we should mention the dagger of the cup and cap for a 1-morphism. These are again drawn topologically as bends in the wire, except now the arrow goes from right to left. This `left' duality obeys the analogous snake equations:
\begin{align*}
\includegraphics[valign=c]{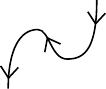}
~~=~~
\includegraphics[valign=c]{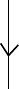}
&&
\includegraphics[valign=c]{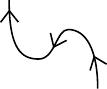}
~~=~~
\includegraphics[valign=c]{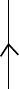}
\end{align*}
From now on we will call these daggers the `left' cup and cap and the cup and cap we defined before the `right' cup and cap. 

\paragraph{Transpose.} 

Let $H$ and $K$ be 1-morphisms of the same type, and let $f: H \to K$ be a 2-morphism. The \emph{transpose} of $f$ is a 2-morphism 
$$f^T: K^* \to H^*.$$
The box for $f^T$ is labelled with $f$, but has the red square in the top left corner. 

The transpose is defined using the right, or equivalently, the left cup and cap, as follows:
\begin{align*}
\includegraphics[valign=c]{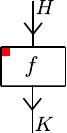}
~~:=~~
\includegraphics[valign=c]{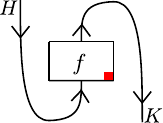}
~~=~~
\includegraphics[valign=c]{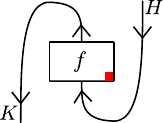}
\end{align*}
For composite 1-morphisms, the transpose can equivalently be defined using the composite cup and cap described above.

\paragraph{Conjugate.}

Let $H$ and $K$ be 1-morphisms of the same type, and let $f: H \to K$ be a 2-morphism. The \emph{conjugate} of $f$ is a 2-morphism
$$
f^*: H^* \to K^*.
$$
The box for $f^T$ is labelled with $f$, but has the red square in the bottom left corner. 

It is defined by $f^* := (f^T)^{\dagger} = (f^{\dagger})^T$; that is, either as the dagger of the transpose or, equivalently, the transpose of the dagger:
\begin{align*}
\includegraphics[valign=c]{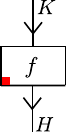}
~~:=~~
\includegraphics[valign=c]{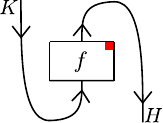}
~~=~~
\left(\includegraphics[valign=c]{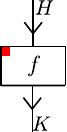}\right)^{\dagger}
\end{align*}

\paragraph{Sliding equations.} With the definitions of dagger, transpose and conjugate we have given, the 2-morphism boxes slide around cups and caps as one might expect from the topology:
\begin{align*}
\includegraphics[valign=c,scale=0.8]{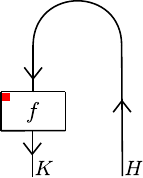}
~=~
\includegraphics[valign=c,scale=0.8]{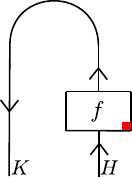}
&&
\includegraphics[valign=c,scale=0.8]{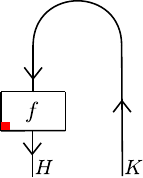}
~=~
\includegraphics[valign=c,scale=0.8]{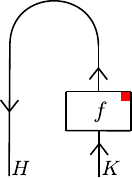}
\end{align*}
\begin{align}\label{eq:sliding}
\includegraphics[valign=c,scale=0.8]{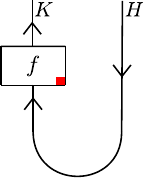}
~=~
\includegraphics[valign=c,scale=0.8]{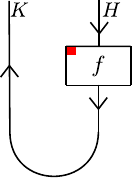}
&&
\includegraphics[valign=c,scale=0.8]{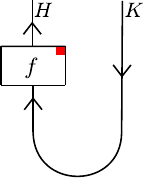}
~=~
\includegraphics[valign=c,scale=0.8]{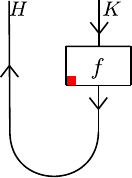}
\end{align}
The analogous equations for the right cup and cap are also obeyed. The analogous equations hold for the composite cups and caps defined for composites of  1-morphisms. We have here ignored the planar regions; they can easily be added in, as we did in~\eqref{eq:cupcap3d}.

\subsubsection{Isometries, unitaries, projections and partial isometries.}

Let $H,K$ be 1-morphisms of the same type. We say that a 2-morphism $f: H \to K$ is:
\begin{itemize}
    \item \emph{Left invertible} if there exists a \emph{left inverse}, that is a 2-morphism $g: K \to H$ such that $g \circ f = \id_H$.
    \item \emph{Invertible} if  there exists an \emph{inverse}, that is a 2-morphism $f^{-1}: K \to H$ such that $f^{-1} \circ f = \id_H$ and $f \circ f^{-1} = \id_{K}$.
    \item An \emph{isometry} if $f^{\dagger}\circ f = \id_H$.
    \item A \emph{coisometry} if $f \circ f^{\dagger} = \id_K$.
    \item A \emph{unitary} if it is both an isometry and a coisometry (in this case the dagger is the inverse).
    \item A \emph{partial isometry} if $(f^{\dagger} \circ f) \circ (f^{\dagger} \circ f)  = f^{\dagger} \circ f$; or, equivalently, if $(f \circ f^{\dagger}) \circ (f \circ f^{\dagger}) = f \circ f^{\dagger}$.
\end{itemize}
We further say that a 2-morphism $f: H \to H$ is \emph{positive} if there exists a 1-morphism $K$ of the same type as $H$ and a 2-morphism $g: H \to K$ such that $f = g^{\dagger}\circ g$.
\subsubsection{Duality for objects.}
Finally, we will briefly discuss some aspects of the duality structure on objects. 

Let $[n]$ be a planar region. There is a 1-morphism $B_{[n]}: [1] \to [n] \boxtimes [n]$ representing a family of Hilbert spaces $
\{(B_{[n]})_{n_1n_2}\}_{n_1,n_2 \in [n] \times [n]}$,
defined by
$$
(B_{[n]})_{n_1n_2} := \delta_{n_1,n_2} \mathbb{C}.
$$
That is, the Hilbert space is zero-dimensional if $n_1 \neq n_2$ and one-dimensional otherwise. 

We will draw this 1-morphism and its dual $B_{[n]}^*: [n] \boxtimes [n] \to [1]$ using dashed lines, as if they were simply bends in the planar region $[n]$  (recall that the planar region for $[1]$ is not depicted in the graphical calculus):
\begin{align*}
\includegraphics[valign=c]{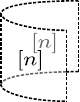}
&&
\includegraphics[valign=c]{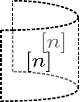}
\end{align*}
The right cap and cup for this duality are depicted as follows:
\begin{align*}
\includegraphics[valign=c]{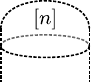}
&&
\includegraphics[valign=c]{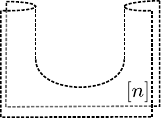}
\end{align*}
Note that these are both critical points of the surface spanned by the planar region $[n]$: the right cap is a local maximum, whereas the right cup is a saddle. The snake equations for this cup and cap can be understood as isotopies of the surface. The left cup and cap are depicted similarly.

There are also unitary 2-morphisms called \emph{cusps} which realise the snake equations in the horizontal direction. The following diagrams depict a cusp and its dagger (note we represent the cusp 2-morphism by a black vertex rather than a box):
\begin{align}\label{eq:cusps}
\includegraphics[valign=c]{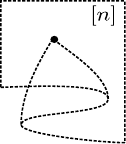}
&&
\includegraphics[valign=c]{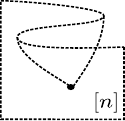}
\end{align}
For a definition of the cusps, see e.g.~\cite[\S{}8.2.8]{Heunen2019}; here we will only use the fact that they exist and are unitary. 

\subsection{Stinespring's and Choi's theorems}
\label{sec:stinespringchoi}
\subsubsection{Endomorphism $C^*$-algebras}

Let $H: O_1 \to O_2$ be a 1-morphism. Let us consider the complex vector space $\End{H}$ of 2-morphisms $H \to H$, where addition and scalar multiplication of indexed families of linear maps are performed in the obvious way. These 2-morphisms form an algebra under vertical composition. They also possess an involution given by the dagger.

With this data, $\End{H}$ is in fact a finite-dimensional $C^*$-algebra. Indeed, all finite-dimensional $C^*$-algebras can be obtained in this way. 
\begin{lem} 
All finite-dimensional $C^*$-algebras can be expressed as endomorphism algebras of 1-morphisms of type $[1] \to [n]$, $n \in \mathbb{N}$.
\end{lem}
\begin{proof}
Recall that every finite-dimensional $C^*$-algebra is a multimatrix algebra 
$$
\bigoplus_{i \in [n]} B(H_i)
$$
for some finite $n \in \mathbb{N}$ and finite-dimensional Hilbert spaces $\{H_{i}\}_{i \in I}$.
Now consider the wire
$H: [1] \to [n]$ representing the family of Hilbert spaces 
$$
\{H_i\}_{i \in [n]}.
$$
It is straightforwardly seen that there is a $*$-isomorphism
\begin{align}\label{eq:endiso}
\End{H} \cong \bigoplus_{i \in [n]} B(H_i).
\end{align}
\end{proof}
\noindent
\begin{rmk}[Positivity]
Positivity of an element $f \in \End{H}$ in the $C^*$-algebraic sense is the same as positivity of that element as a 2-morphism in $\TwoFHilb$, as defined in Section~\ref{sec:diagcalc}. 
\end{rmk}
\begin{rmk}[Traces and dimension]
We remark that under the isomorphism~\eqref{eq:endiso} the usual matrix trace $\Tr(x)$ of an element $x \in \bigoplus_{i \in [n]} B(H_i)$ can be expressed as the following 2-morphism:
\begin{align*}
\includegraphics[valign=c,scale=0.8]{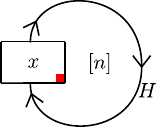}
\end{align*}
To understand how this 2-morphism can represent the trace, which is a scalar, observe that the surrounding planar region is invisible and therefore labelled by $[1]$, and the source and target are empty, so it represents a single linear map $\mathbb{C} \to \mathbb{C}$, which is precisely a scalar. From now on we will simply refer to such 2-morphisms as scalars without repeating this explanation.

We will later be particularly interested in the scalar $\dim(H) := \Tr(\id_H)$:
\begin{align*}
\includegraphics[valign=c,scale=0.8]{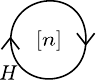}
\end{align*}
If $H$ represents a family of Hilbert spaces $\{H_i\}_{i \in [n]}$, we have $\dim(H) = \sum_{i \in [n]} \dim(H_i)$; it is therefore a positive, natural number.
\end{rmk}

\subsubsection{Generalised Stinespring's theorem}
For a more sophisticated statement of this theorem which applies to general von Neumann algebras, see~\cite{Allen2023}.
\begin{thm}[{\cite[Thm. 1.2]{Allen2023}}]\label{thm:stinespring}
Let $O_1$, $O_2$ be objects, and 
$H: [1] \to O_1$ and $K: [1] \to O_2$ be 1-morphisms.

Let $\mathcal{F}: \End{H} \to \End{K}$ be a completely positive (CP) map of finite-dimensional $C^*$-algebras. Then there exists a 1-morphism $E: O_1 \to O_2$, which we call an \emph{environment}, and a 2-morphism $V: K \to H \otimes E$, which we call a \emph{dilation}, such that 
$$
\mathcal{F}(x) =  V^{\dagger} \circ (x \otimes \id_E) \circ V;
$$
or, as a 2-morphism diagram:
\begin{align*}
\includegraphics[valign=c]{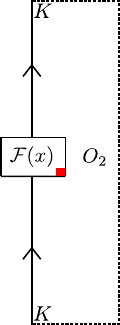}~~=~~\includegraphics[valign=c]{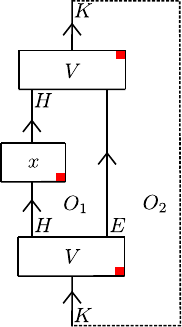}
\end{align*}
We call the pair $(E,V)$ a \emph{Stinespring representation} of the CP map $\mathcal{F}$. The CP map $\mathcal{F}$ is unital (i.e. identity-preserving) iff $V$ is an isometry.  Two representations $(E,V)$ and $(E',V')$ of the same CP map are related by a partial isometry on the environment: that is, there exists a partial-isometric 2-morphism $\sigma: E \to E'$ such that
\begin{align}\label{eq:repmorphism}
V' = (\id_{H} \otimes \sigma) \circ V && V = (\id_{H} \otimes \sigma^{\dagger}) \circ V'.
\end{align}
Indeed, every CP map $\mathcal{F}: \End{ H} \to \End{K}$ has a \emph{minimal representation}. This minimal representation is unique up to a unitary 2-morphism on the environment and is related to every other representation by an isometric 2-morphism on the environment. A representation $(E,V)$ is minimal iff the following 2-morphism $E \to E$ is invertible:
\begin{align}\label{eq:minimalitycond}
\includegraphics[valign=c]{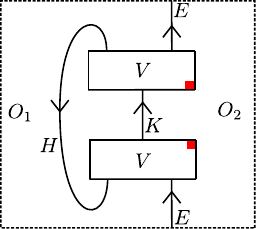}
\end{align}
\end{thm}
\noindent
The formulation of the generalised Stinespring's theorem we have given is in the Heisenberg picture, where dynamical maps are unital CP maps. In order to move to the Schr\"{o}dinger picture more commonly used in quantum information, where dynamical maps are \emph{trace-preserving} CP maps, we define the dual CP map. 
\begin{defn}\label{def:dualcp}
    Let $O_1, O_2$ be objects, and let $H: [1] \to O_1$ and $K: [1] \to O_2$ be 1-morphisms. Let $\mathcal{F}: \End{H} \to \End{K}$ be a CP map and let $(E,V)$ be a Stinespring representation of $\mathcal{F}$. The dual CP map $\mathcal{F}_*:\End{K} \to \End{H}$ is defined as follows:
    \begin{align*}
    \includegraphics[valign=c]{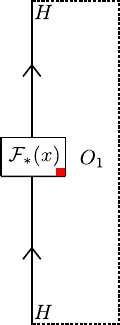}
    ~~:=~~
    \includegraphics[valign=c]{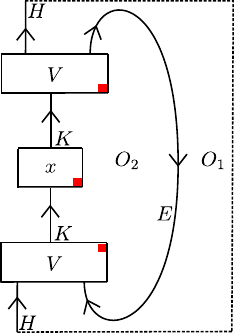}
    \end{align*}
    That is, $\mathcal{F}_*$ is a CP map with a Stinespring representation $(E^*,V_*)$, where $V_*: H \to K \otimes E^*$ is defined as follows:
    \begin{align*}
    \includegraphics[valign=c]{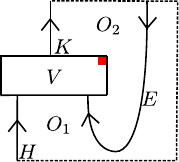}
    \end{align*}
    The definition of the dual CP map $\mathcal{F}_*$ is independent of the choice of Stinespring representation $(V,E)$ of $\mathcal{F}$, by uniqueness of the Stinespring representation up to partial isometry~\eqref{eq:repmorphism}. The dual CP map $\mathcal{F}_*$ is trace preserving iff $\mathcal{F}$ is unital iff $V: K \to H \otimes E$ is an isometry. The dual of the dual CP map is the original CP map (we invite the reader to prove this  using a snake equation~\eqref{eq:snakes}).
\end{defn}

\subsubsection{Choi's theorem}

As a corollary of Stinespring's theorem, completely positive maps can be identified with positive elements in a certain f.d. $C^*$-algebra.
\begin{cor}[Choi's theorem]
Let $H: [1] \to [m]$ and $K: [1] \to [m]$ be 1-morphisms. 
There is a bijective correspondence (indeed, an isomorphism of convex cones) between CP maps $\mathcal{F}:\End{H} \to \End{K}$ and positive elements $\widetilde{F}$ in the f.d. $C^*$-algebra $\End{K \boxtimes \hat{H}}$, where $\hat{H}$ is the following 1-morphism:
\begin{align}\label{eq:hath}  
\includegraphics[valign=c]{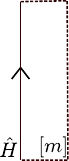}~~:=~~\includegraphics[valign=c]{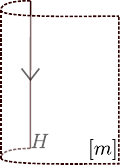}
\end{align}
This correspondence is as follows. Take any representation $(E,V)$ of $\mathcal{F}$. Then the corresponding positive element $\widetilde{\mathcal{F}} \in \End{K \boxtimes \hat{H}}$ is given by:
\begin{align}\label{eq:choi}
\includegraphics[valign=c]{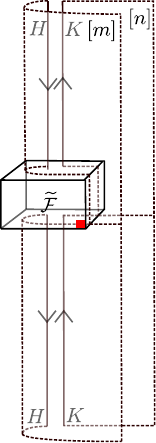}~~:=~~\includegraphics[valign=c]{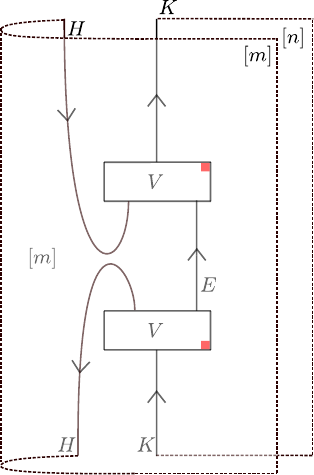}
\end{align}
This positive element is independent of the choice of representation, by~\eqref{eq:repmorphism}. The CP map is trace-preserving iff the positive element $\widetilde{\mathcal{F}} \in \End{K \boxtimes \hat{H}}$ obeys the following condition:
\begin{align}\label{eq:tpcond}
    \includegraphics[valign=c]{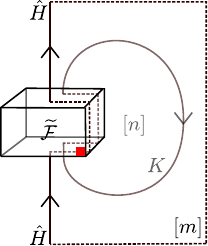}~~=~~\includegraphics[valign=c]{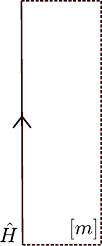}
\end{align}
Here we have used the right cap and left cup to trace out the $K$-wire on the LHS. The RHS depicts the identity 2-morphism on $\hat{H}$; recall that in the graphical calculus the identity 2-morphism is invisible. In inline notation, we write $\Tr_{K} (\Ftilde) = \id_{\hat{H}}$.
\end{cor}
\begin{notation}
From now on we will use this tilde notation to indicate the use of Choi's theorem; if $\mathcal{F}: \End{H} \to \End{K}$ is a CP map, then $\widetilde{F} \in \End{K \boxtimes \hat{H}}$ is the associated positive operator, which we will sometimes call the \emph{Choi operator} of the CP map. 
\end{notation}
\begin{rmk}\label{rem:recoverchan}
To recover a representation of the CP map $\mathcal{F}$ from the positive element $\Ftilde \in \End{K \boxtimes \hat{H}}$, we have the following procedure. We begin with the element $\Ftilde$:
\begin{align*}
\includegraphics[valign=c]{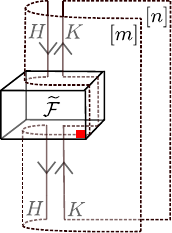}
\end{align*}
We now trace out the bend 1-morphism using its left cap and right cup (remember these are saddles) to obtain the following 2-morphism:
\begin{align*}
\includegraphics[valign=c]{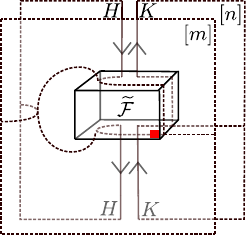}
\end{align*}
Finally, we close off the front plane by introducing a bend and its dual, together with a right cap and a left cup:
\begin{align}\label{eq:recovertechnique}
\includegraphics[valign=c]{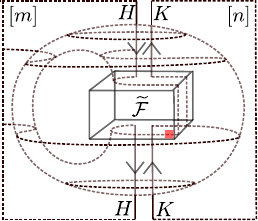}
\end{align}
This is a 2-morphism in $\End{H^* \otimes K}$. It is positive since $\Ftilde$ is, and therefore splits as $V^{\dagger} \circ V$ for some 1-morphism $E: [m] \to [n]$ and 2-morphism $V: H^* \otimes K \to E$:
\begin{align}\label{eq:recoverv}
\includegraphics[valign=c]{graphics/StinespringAndChoi/recover3.pdf}
=
\includegraphics[valign=c]{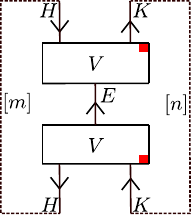}
\end{align}
The CP map $\mathcal{F}$ then has a representation with environment $E$ and the following dilation $K \to H \otimes E$:
\begin{align}\label{eq:recovertranspose}
\includegraphics[valign=c]{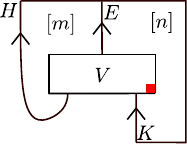}
\end{align}
\end{rmk}
\begin{rmk}
This algebra $\End{K \boxtimes \hat{H}}$ is what we called $\underline{\mathrm{Hom}}(\End{H},\End{K})$ in Definition~\ref{def:detsupermap}. The condition~\eqref{eq:tpcond} is what we called (TP) in Definition~\ref{def:detsupermap}.
\end{rmk}
\begin{rmk}
There are other isomorphic endomorphism algebras we could use for $\underline{\mathrm{Hom}}(\End{H},\End{K})$; for example, in~\cite[Thm. 4.13]{Verdon2022} the algebra $\End{H^* \otimes K}$ was used. The choice we have made here is more convenient for using Stinespring's theorem to dilate supermaps, since the 1-morphism $K \boxtimes \hat{H}$ has source object $[1]$. 
\end{rmk}

\section{Proof of realisation theorem}\label{sec:proof}

\begin{lem}[{C.f.~\cite[Lemma 1]{CDP08}}]\label{lem:1}
Let $H: [1] \to [m]$ and $K: [1] \to [n]$ be 1-morphisms, and let $C \in \End{K \boxtimes \hat{H}}$, where $\hat{H}$ is defined as in~\eqref{eq:hath}. If $\Tr[C \circ \Ftilde] = 1$ for all positive elements $\Ftilde \in \End{K \boxtimes \hat{H}}$ obeying the condition (TP) shown in~\eqref{eq:tpcond}, then $$C = \id_{K} \boxtimes \rho$$ for some $\rho \in \End{\hat{H}}$ such that $\Tr[\rho] = 1$. Or, in diagrams:
\begin{align*}
&\includegraphics[valign=c]{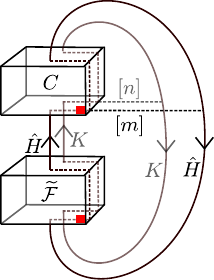} ~~=~~1~~\forall ~ \Ftilde \ge 0 :~~ \includegraphics[valign=c]{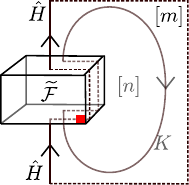}~~=~~\includegraphics[valign=c]{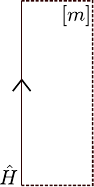}
\\
&\implies~~\includegraphics[valign=c]{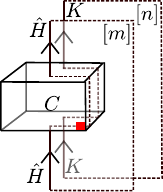}~~=~~
\includegraphics[valign=c]{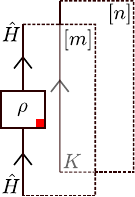}~,~~~~\includegraphics[valign=c]{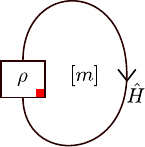}~~=~~1.
\end{align*}
Additionally, it follows that if $C \ge 0$, then $\rho \ge 0$. 
\end{lem}
\begin{proof}
Choose any positive $\sigma \in \End{K}$ such that $\Tr[\sigma] = 1$: 
     \begin{align*}
          \includegraphics[valign=c]{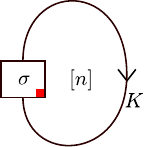}
          ~~=~~1
\end{align*}
    Given any positive element $x \in \End{K \boxtimes \hat{H}}$ (not necessarily obeying the condition (TP)), define the following positive element $\Ftilde \in \End{K \boxtimes \hat{H}}$: 
    \begin{align*}
    \includegraphics[valign=c]{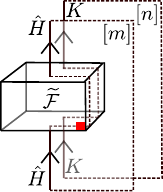}
    ~~:=~~&
    \includegraphics[valign=c]{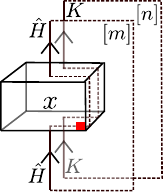}
    ~~+~~
    \includegraphics[valign=c]{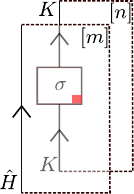}\\&-~~
    \includegraphics[valign=c]{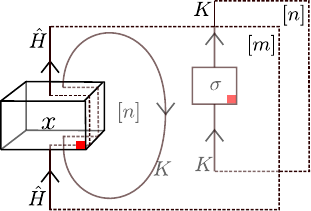}
\end{align*}
It is straightforwardly seen that this positive element $\Ftilde$ obeys the (TP) condition~\eqref{eq:tpcond}:
\begin{align*}
    \includegraphics[valign=c]{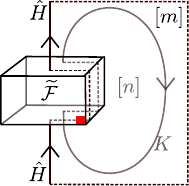}
    ~~=~~&
    \includegraphics[valign=c]{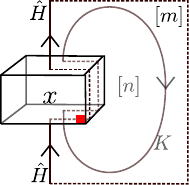}
    ~~+~~
    \includegraphics[valign=c]{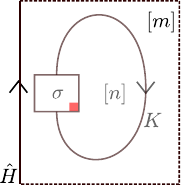}
    \\&-~~
    \includegraphics[valign=c]{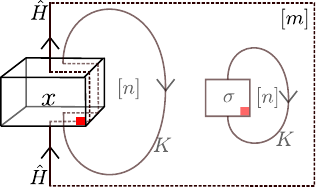}
\end{align*}
\begin{align*}
    &=~~\includegraphics[valign=c]{graphics/theorem1/xtrace.pdf}
    ~~+~~
    \includegraphics[valign=c]{graphics/theorem1/idh.pdf}
    ~~-~~
    \includegraphics[valign=c]{graphics/theorem1/xtrace.pdf}
\end{align*}
\begin{align*}
    =~~
    \includegraphics[valign=c]{graphics/theorem1/idh.pdf}
\end{align*}
Here for the second equality we used that $\Tr(\sigma)=1$.

Composing $\Ftilde$ with $C$ and taking the trace, we obtain (by definition of $\Ftilde$):
\begin{align}\nonumber
\includegraphics[valign=c,scale=1]{graphics/theorem1/cftrace.pdf}
~~&=~~
\includegraphics[valign=c,scale=1]{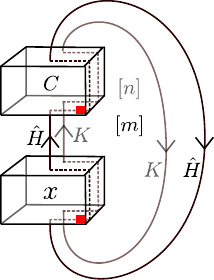}
~~+~~
\includegraphics[valign=c,scale=1]{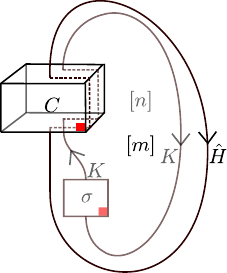}
\\\label{eq:ftildedef}&-~~
\includegraphics[valign=c,scale=1]{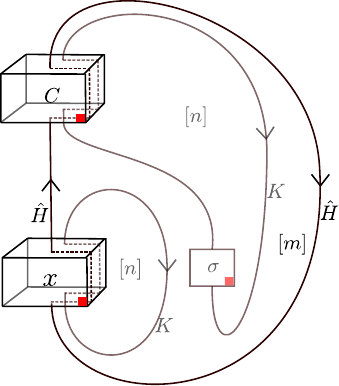}
\end{align}
$\Ftilde$ and $ \sigma \boxtimes \id_{\hat{H}}$ both obey the (TP) condition~\eqref{eq:tpcond}, so, by the assumption on $C$ in the statement of the theorem: 
\begin{align*}
\includegraphics[valign=c,scale=1]{graphics/theorem1/cftrace.pdf}
~~=~~1~~=~~
\includegraphics[valign=c,scale=1]{graphics/theorem1/csigmaboxidtrace.pdf}
\end{align*}
Substituting this into~\eqref{eq:ftildedef} yields:
\begin{align}\label{eq:c}
\includegraphics[valign=c]{graphics/theorem1/cxtrace.pdf}
~~=~~\includegraphics[valign=c]{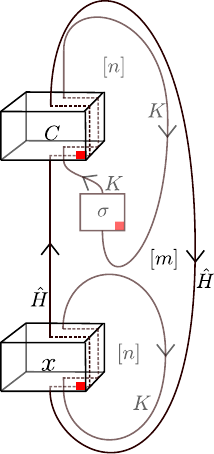}
\end{align}
The above is true for all $x \ge 0$; hence, since every element of $\End{K \boxtimes \hat{H}}$ is a linear combination of positive elements, it is true for all $x \in \End{K \boxtimes \hat{H}}$. Observe that the sesquilinear form $\braket{x,y} := \Tr(x^{\dagger} y)$ defines an isomorphism $\phi: \End{K \boxtimes \hat{H}} \overset{\sim}{\to} \overline{\End{K \boxtimes \hat{H}}^*}$ by $\phi(x)(y) = \Tr(x^{\dagger} y)$. Then~\eqref{eq:c} implies the equality
\begin{align}\label{eq:csplits}
\includegraphics[valign=c]{graphics/theorem1/c.pdf}~~=~~\includegraphics[valign=c]{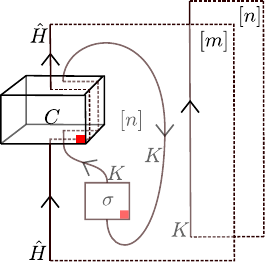},
\end{align}
since these elements have the same image under the isomorphism $\phi$.
This is the decomposition of $C$ in the theorem statement, where 
\begin{align*}
\includegraphics[valign=c]{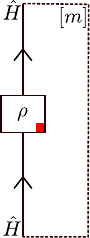}~~:=~~\includegraphics[valign=c]{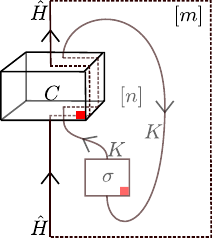},
\end{align*}
In order to see that $\rho$ is positive if $C$ is, observe from~\eqref{eq:csplits} that 
\begin{align*}
    \includegraphics[valign=c]{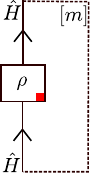}~~=~~\frac{1}{\dim(K)}~~\includegraphics[valign=c]{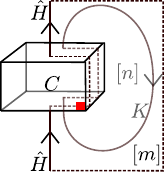}~,
\end{align*}
which is clearly the composition of a 2-morphism with its dagger, since $C$ is positive. 
\end{proof}

\begin{defn}[Formal restatement of Definition~\ref{def:detsupermap}]\label{def:detsupermapformal}
Let \begin{align*}
\Hin: [1] \to [m], && \Hout: [1] \to [n], && \Kin: [1] \to [m'], && \Kout: [1] \to [n']
\end{align*}
be 1-morphisms. Recall from Section~\ref{sec:stinespringchoi} that each completely positive map $\mathcal{F}: \End{\Hin} \ra \End{\Hout}$ corresponds to a positive element $\widetilde{\mathcal{F}}: \End{\Hout \boxtimes \widehat{\Hin}}$, called its Choi operator (defined in~\eqref{eq:choi}).

A \emph{deterministic supermap} from channels (i.e. completely positive trace-preserving maps) of type $\End{\Hin} \to \End{\Hout}$ to channels of type $\End{\Kin} \to \End{\Kout}$ is a completely positive map \begin{equation*}
\Stilde : \End{\Hout \boxtimes \widehat{\Hin}} \ra \End{\Kout \boxtimes \widehat{\Kin}}
\end{equation*}
which maps Choi operators obeying the (TP) condition~\eqref{eq:tpcond} to Choi operators obeying the (TP) condition:
\begin{equation}\label{eq:tpp}
    \Tr_{\Hout}[\Ftilde] = \id_{\Hin^*} \implies \Tr_{\Kout}[\Stilde(\Ftilde)] = \id_{\Kin^*}.
\end{equation}
\end{defn}
\noindent
We use Stinespring's theorem (Theorem~\ref{thm:stinespring}) to obtain a dilation $(E,S)$ for the CP map $\mathcal{S}$. Using this dilation the condition~\eqref{eq:tpp} can be written diagrammatically as follows:
\begin{align*}
\includegraphics[valign=c]{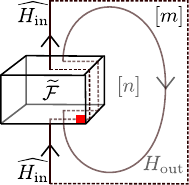}~~=~~
\includegraphics[valign=c]{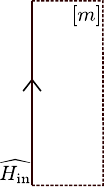} \end{align*}
\begin{align*}
\implies~~ 
\includegraphics[valign=c]{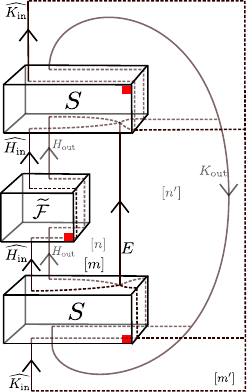}~~=~~
\includegraphics[valign=c]{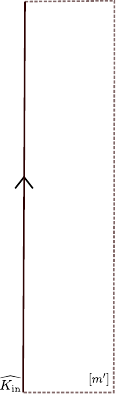}
\end{align*}
For the next lemma, recall the definition of the dual of a CP map (Definition~\ref{def:dualcp}).
\begin{lem}[{C.f.~\cite[Lemma 2]{CDP08}}]\label{lem:2}
Let 
\begin{align*}
\Hin: [1] \to [m],&& \Hout: [1] \to [n],&& \Kin: [1] \to [m'],&& \Kout: [1] \to [n']
\end{align*}
be 1-morphisms, and let $\Stilde : \End{\Hout \boxtimes \widehat{\Hin}} \ra \End{\Kout \boxtimes \widehat{\Kin}}$ be a deterministic supermap. 

Then there exists a unital CP map $\mathcal{N}: \End{\widehat{\Hin}} \ra \End{\widehat{\Kin}}$ such that, for any positive $\rho \in \End{\widehat{\Kin}}$ such that $\Tr(\rho) = 1$:
\begin{equation}
    \Stilde_* (\id_{\Kout} \boxtimes \rho) 
    = \id_{\Hout} \boxtimes \mathcal{N}_*(\rho) ,
\end{equation}
or, in diagrams (where $(E_N,N)$ is a minimal Stinespring representation of the unital CP map $\mathcal{N}$):
\begin{align*}
\includegraphics[valign=c]{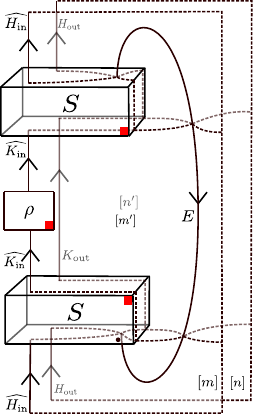}
~~=~~
\includegraphics[valign=c]{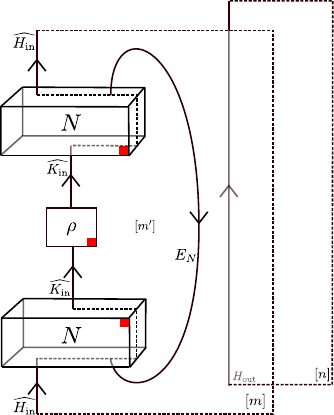}
\end{align*}
\end{lem}
\begin{proof}
Let $\rho \in \End{\widehat{\Kin}}$ be any positive element satisfying $\Tr(\rho) = 1$. Now let $C := \id_{\Kout} \boxtimes \rho \in \End{\Kout \boxtimes \widehat{\Kin}}$. Clearly, $C$ is also positive. 
\begin{align*}
\includegraphics[valign=c]{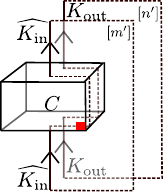}~~:=~~\includegraphics[valign=c]{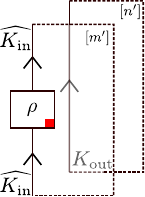}
\end{align*}
Then, for any positive $\Ftilde \in \End{\Hout \boxtimes \widehat{\Hin}}$ satisfying the (TP) condition~\eqref{eq:tpcond}, we have 
    \begin{equation}\label{eq:lem21}
        1 = \Tr[\Stilde_* (C) \circ \Ftilde],
    \end{equation}
which is seen by the following equation:
\begin{align*}1~~=~~\includegraphics[valign=c,scale=1]{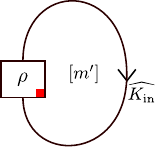}
~~=~~
\includegraphics[valign=c,scale=1]{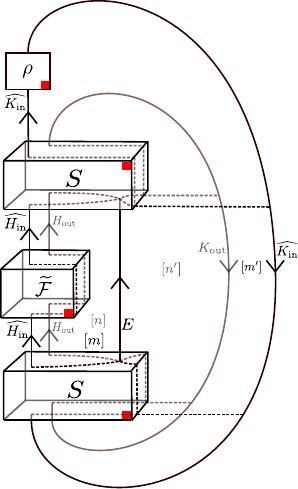}
\end{align*}
\begin{align}\label{eq:firstsliding}
=\includegraphics[valign=c,scale=1]{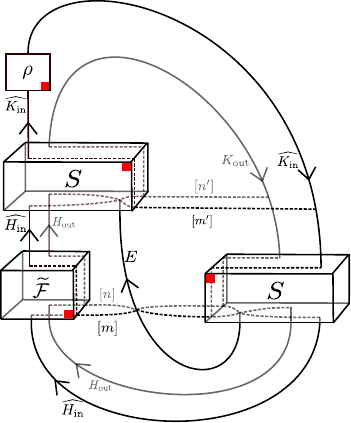}
~~=~~
\includegraphics[valign=c,scale=1]{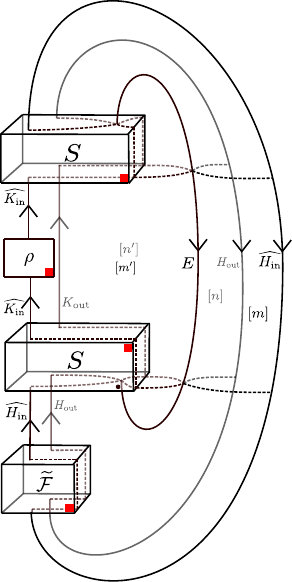}
\end{align}
Here the first equality is by $\Tr(\rho) = 1$; the second equality is by the fact that $\mathcal{S}$ preserves the (TP) condition~\eqref{eq:tpcond}; the third equality is by pulling the bottom $S$-box around the cup using the sliding equations~\eqref{eq:sliding}; and the final equality is by sliding that transposed $S$-box around the cap using the sliding equations~\eqref{eq:sliding}. In the last diagram we indeed see $\Tr[\mathcal{S}_*(C) \circ \Ftilde]$ and so~\eqref{eq:lem21} is proved.
We now apply Lemma~\ref{lem:1} to obtain
    \begin{equation}
        \Stilde_* (\id_{\Kout} \boxtimes \rho) = \id_{\Hout} \boxtimes \sigma,
    \end{equation}
    where $\Tr[\sigma] = 1$ and $\sigma \geq 0$. 

As $\rho \mapsto \rho \ot \id_{\Kout}$, $\Stilde_*$ and $ \sigma \ot \id_{\Hout} \mapsto \sigma$ are all CP, it follows that $\sigma = \mathcal{N}_*(\rho)$, where $\mathcal{N}_* : \End{\widehat{\Kin}} \ra \End{\widehat{\Hin}}$ is a CP map defined as follows:
\begin{align*}
\includegraphics[valign=c,scale=1]{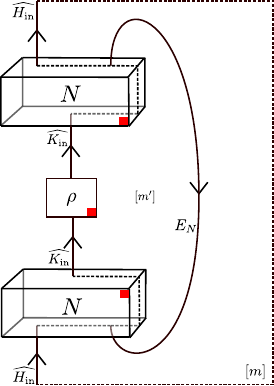}
~~:=~~
\frac{1}{\dim(\Hout)}~
\includegraphics[valign=c,scale=1]{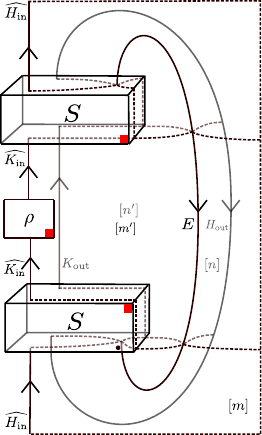}
\end{align*}
To see that $\mathcal{N}_*$ is trace preserving it is sufficient to note that it preserves the trace on positive elements of trace 1; it therefore preserves the trace on all elements of $\End{\widehat{\Kin}}$ by linearity of the trace. Since $\mathcal{N}_*$ is trace-preserving, $\mathcal{N}$ is unital and the result follows.
\end{proof}
\begin{lem}[{C.f.~\cite[Lemma 3]{CDP08}}]\label{lem:3}
Let \begin{align*}
    \Hin: [1] \to [m], && \Hout: [1] \to [n], && \Kin: [1] \to [m'], && \Kout: [1] \to [n']
\end{align*} be 1-morphisms, and let $\Stilde : \End{\Hout \boxtimes \widehat{\Hin}} \ra \End{\Kout \boxtimes \widehat{\Kin}}$ be a deterministic supermap. Then there exists an unital CP map $\mathcal{N} : \End{\widehat{\Hin}} \ra \End{\widehat{\Kin}}$ such that, for any positive $C \in \End{\Hout \boxtimes \widehat{\Hin}}$,
\begin{equation}\label{eq:lemma3}
    \Tr_{\Kout}[\Stilde(C)] = \mathcal{N} (\Tr_{\Hout} [C]),
\end{equation}
or, in diagrams:
\begin{align}\label{eq:lem3}
\includegraphics[valign=c,scale=1]{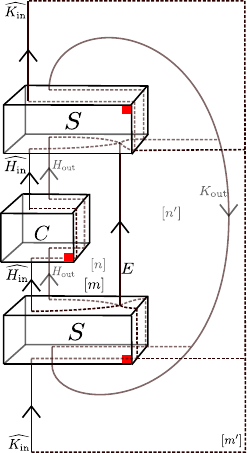}~~=~~\includegraphics[valign=c,scale=1]{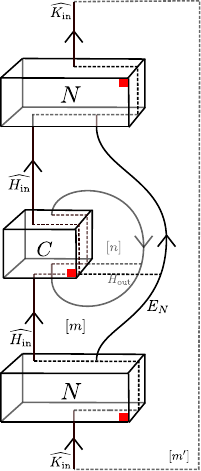}
\end{align}
\end{lem}
\begin{proof}
Let $\rho \in \End{\widehat{\Kin}}$ be any positive element satisfying $\Tr[\rho] = 1$. The following equation shows that $\Tr[\rho \circ \Tr_{\Kout}[\mathcal{S}(C)]] = \Tr[\rho \circ \mathcal{N}(\Tr_{\Hout}[C])]$:
\begin{align*}
\includegraphics[valign=c,scale=1]{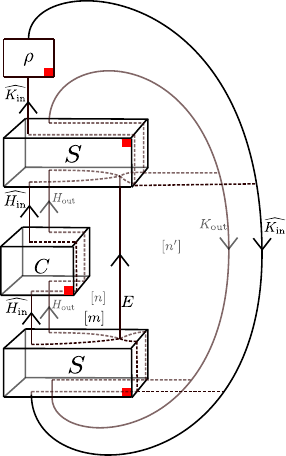}~~=~~
\includegraphics[valign=c,scale=1]{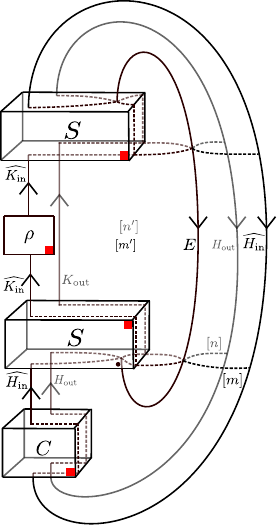}
\end{align*}
\begin{align*}
=~~\includegraphics[valign=c,scale=1]{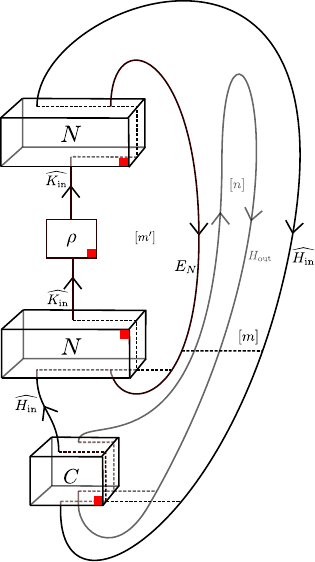}
~~=~~
\includegraphics[valign=c,scale=1]{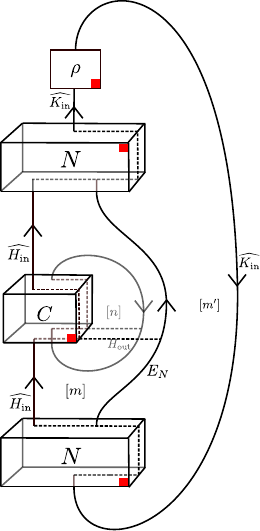}
\end{align*}
Here for the first equality we used the sliding equations~\eqref{eq:sliding} to pull the bottom $S$-box round the loop, as we did in~\eqref{eq:firstsliding}; for the second equality we applied Lemma~\ref{lem:2}; and for the third equality we used the sliding equations~\eqref{eq:sliding} to pull the top $N$-box round the loop. 

By linearity of the trace and decomposition of an arbitrary element of $\End{\widehat{\Kin}}$ as a linear combination of positive elements, it follows that $\Tr[x \circ \Tr_{\Kout}[\mathcal{S}(C)]] = \Tr[x \circ \mathcal{N}(\Tr_{\Hout}[C])]$ for any $x \in \End{\widehat{\Kin}}$. It follows that $\Tr_{\Kout}[\mathcal{S}(C)] = \mathcal{N}(\Tr_{\Hout}[C])$ by a similar argument to that made at the end of the proof of Lemma~\ref{lem:1}.
\end{proof}

\begin{thm}[{C.f.~\cite[Theorem 1]{CDP08}}]\label{thm:1}
    Let 
    \begin{align*}
        \Hin: [1] \to [m], && \Hout: [1] \to [n], && \Kin: [1] \to [m'], && \Kout: [1] \to [n']
    \end{align*}
    be 1-morphisms, and let $\Stilde : \End{\Hout \boxtimes \widehat{\Hin}} \ra \End{\Kout \boxtimes \widehat{\Kin}}$ be a deterministic supermap. Then $\Stilde$ can be decomposed in the following way, where $N,W$ are isometries:
\begin{align}\label{eq:thm1}\includegraphics[valign=c,scale=1]{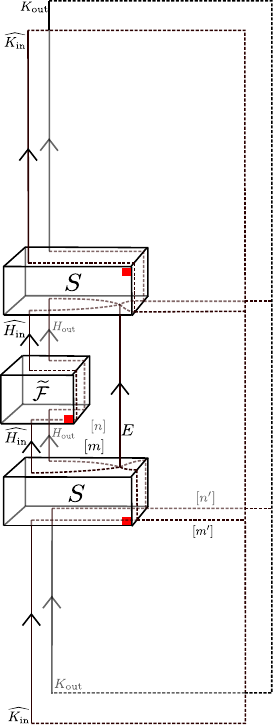}~~=~~\includegraphics[valign=c,scale=1]{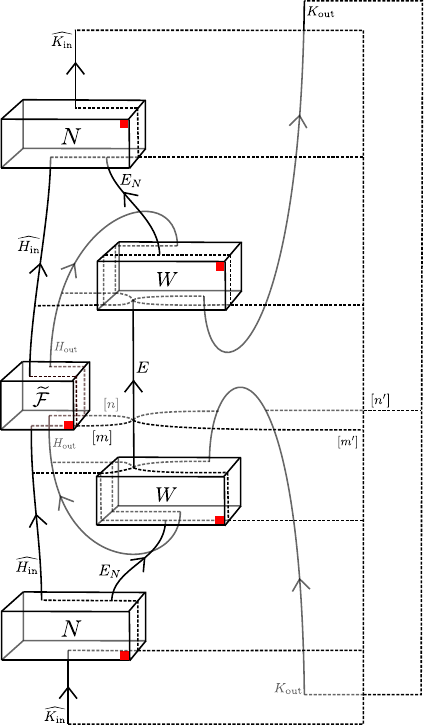}
\end{align}
\end{thm}
\begin{proof}
Looking back at Lemma~\ref{lem:3}, and in particular at~\eqref{eq:lem3}, we see two equal CP maps $\End{\Hout \boxtimes \widehat{\Hin}} \to \End{\widehat{\Kin}}$, whose dilations are as follows:
\begin{align}\label{eq:dilationcompare}
\includegraphics[valign=c,scale=1]{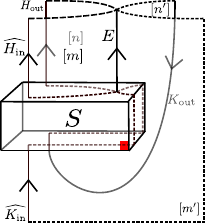},
&&
\includegraphics[valign=c,scale=1]{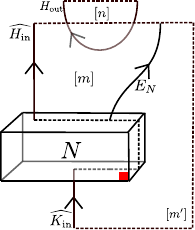}
\end{align}
The dilation on the right is minimal. We can see this by the invertibility condition for minimality~\eqref{eq:minimalitycond}; the dilation on the right is minimal iff 
\begin{align*}
\includegraphics[valign=c,scale=1]{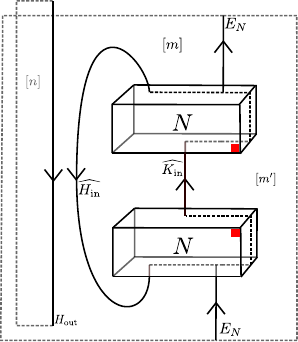}
\end{align*}
is invertible, which follows immediately from the fact that $N$ is a minimal dilation.

By Theorem~\ref{thm:stinespring}, the dilation on the left of~\eqref{eq:dilationcompare} is therefore related to the dilation on the right by an isometry $W$ on the environment:
\begin{align*}
\includegraphics[valign=c,scale=1]{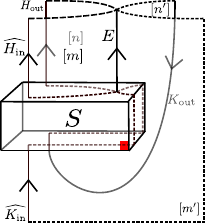}
~~=~~
\includegraphics[valign=c,scale=1]{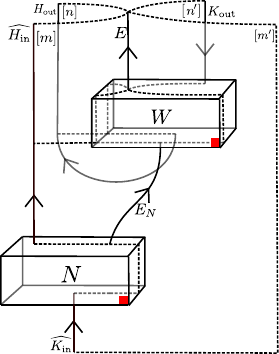}
\end{align*}
Transposing the $\Kout$ wire, we obtain the following expression for $S$:
\begin{align*}
\includegraphics[valign=c,scale=1]{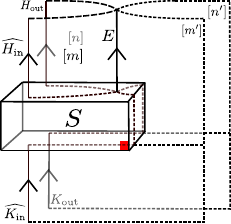}
~~=~~
\includegraphics[valign=c,scale=1]{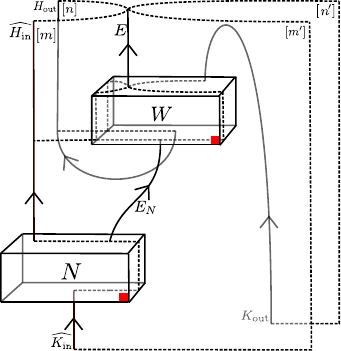}
\end{align*}
Substitution into the LHS of~\eqref{eq:thm1} gives the result.
\end{proof}
\noindent
\paragraph{Final argument.} We now explain how to deduce the realisation theorem  stated in the introduction (Theorem~\ref{thm:mainintro}) from Theorem~\ref{thm:1}. We first input the definition of the 1-morphisms $\widehat{\Hin}$ and $\widehat{\Kin}$, and the Choi operator $\widetilde{\mathcal{F}}$ of a channel $\mathcal{F}: B(\Hin) \to B(\Hout)$ with dilation $(E_F,F)$ (this was stated in~\eqref{eq:choi}), to expand the RHS of~\eqref{eq:thm1}:
\begin{align}\label{eq:derive1}
\includegraphics[valign=c,scale=.95]{graphics/theorem1/statement2.pdf}~~=~~\includegraphics[valign=c,scale=.95]{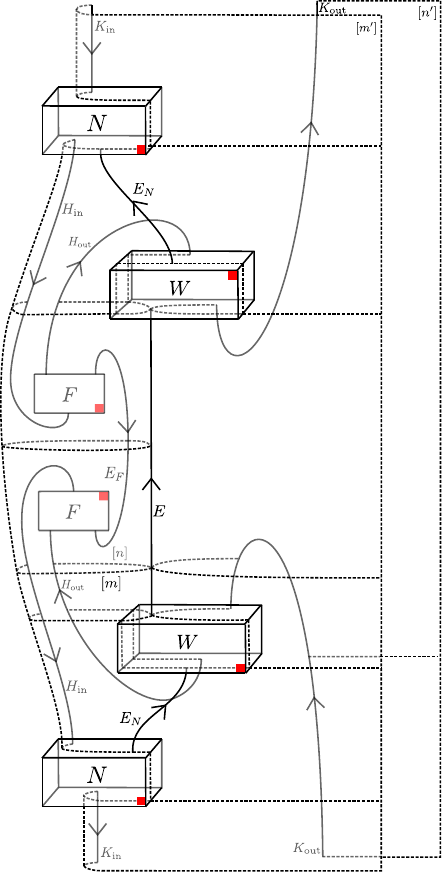}
\end{align}
Since the RHS of~\eqref{eq:derive1} is symmetric under reflection across a horizontal axis, we just focus on the bottom half; we will rewrite the top half in the same way. Using the sliding equations~\eqref{eq:sliding} we have the following series of equalities:
\begin{align*}
    \includegraphics[valign=c,scale=.95]{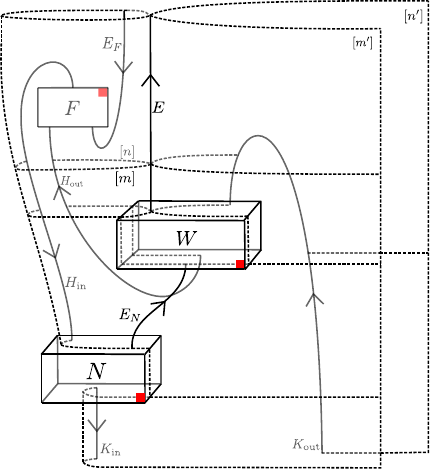}
    ~~=~~
    \includegraphics[valign=c,scale=.95]{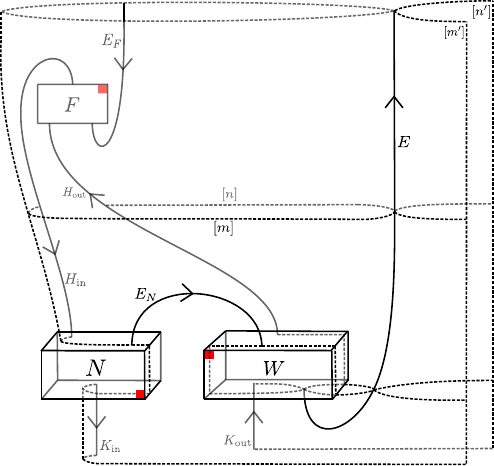}
    \end{align*}
\begin{align*}
   =~~ \includegraphics[valign=c,scale=.95]{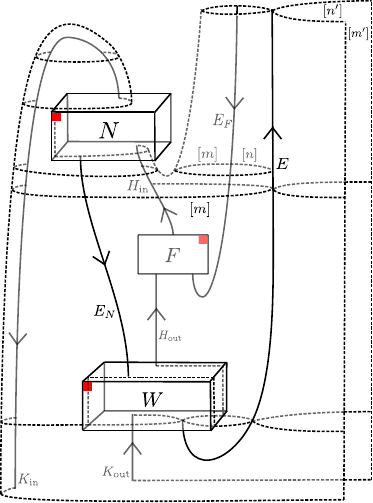}
\end{align*}
We input this rewrite into the RHS of~\eqref{eq:derive1} and then obtain the following additional equality:
\begin{align*}
    \includegraphics[valign=c,scale=.95]{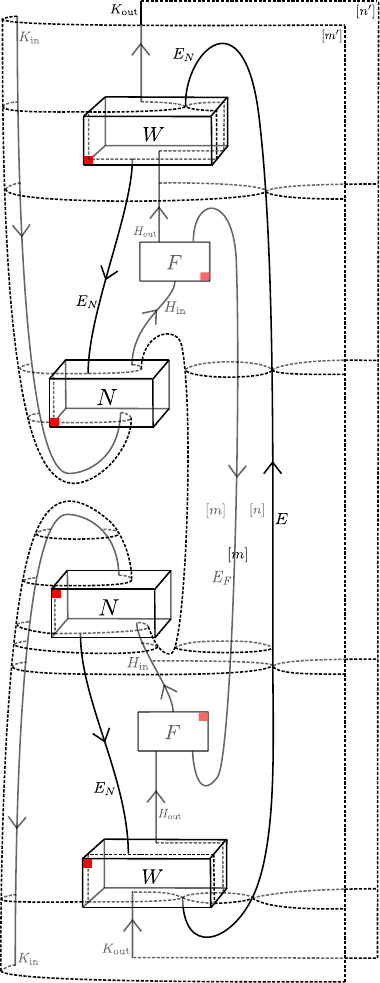}
    ~~=~~
\includegraphics[valign=c,scale=.95]{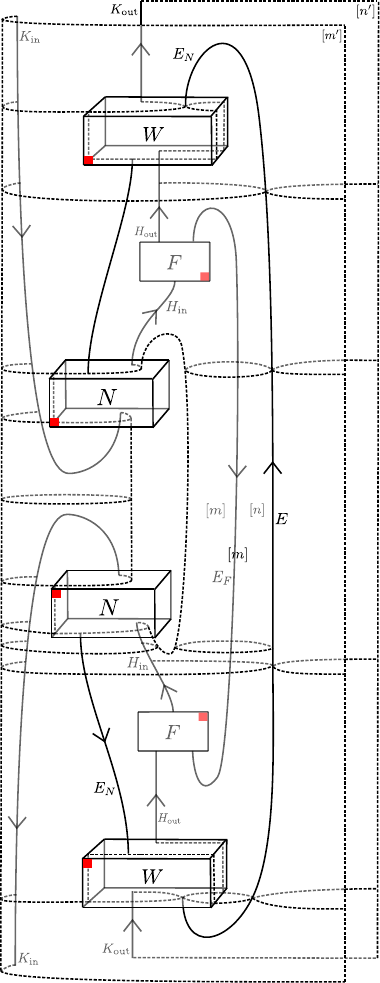}    
\end{align*}
This equality holds because the cup and cap for $B_{[m']}$ in the centre left of the diagram belong to the same planar region, so they can be fused. Now we recall the method given in Remark~\ref{rem:recoverchan} for recovering the dilation of a channel from its Choi operator. Following this procedure up to~\eqref{eq:recovertechnique}, we obtain the following diagram:
\begin{align*}
\includegraphics[valign=c,scale=.95]{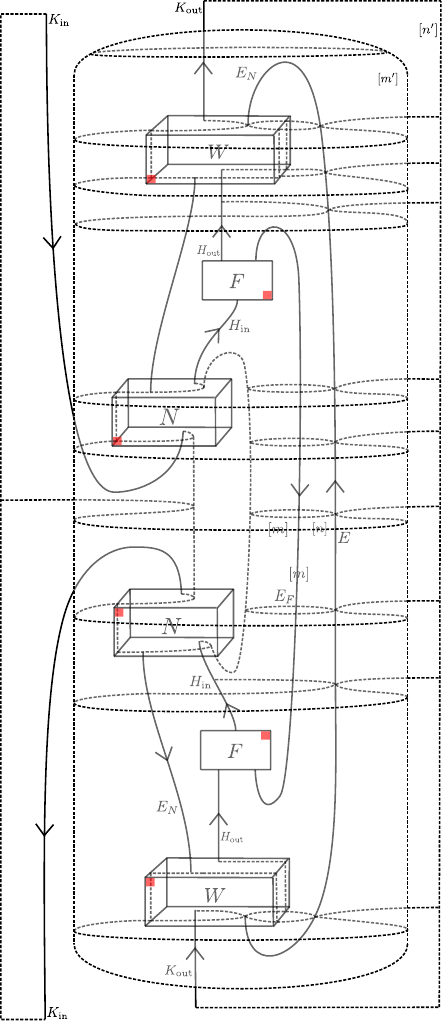}    
\end{align*}
Finally, recalling~\eqref{eq:cusps} we insert a cusp and its inverse in the centre left of the diagram:
\begin{align*}
\includegraphics[valign=c,scale=.95]{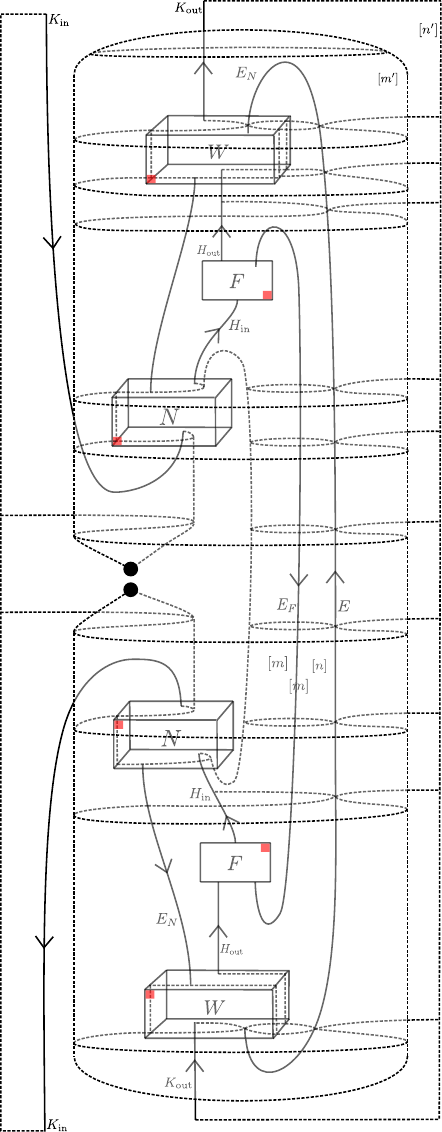}
\end{align*}
This diagram is analogous to the RHS of~\eqref{eq:recoverv}, where the 2-morphism $V$ is defined as follows:
\begin{align*}
\includegraphics[valign=c,scale=.95]{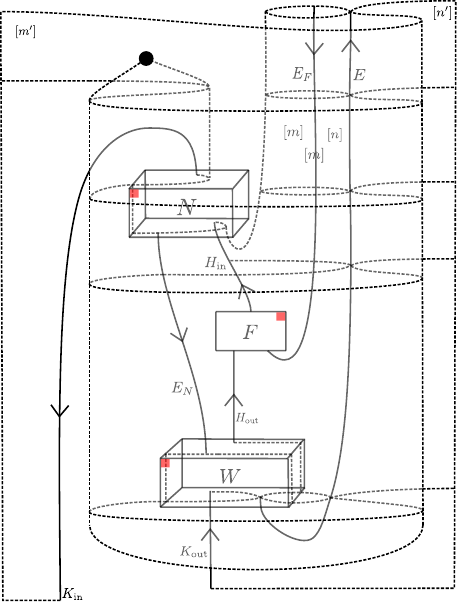}    
\end{align*}
Transposing the $\Kin$ wire as in~\eqref{eq:recovertranspose}, noting that $N^T,F^{\dagger},W^{T}$ are coisometries, and recalling the definition of a dual CP map (Definition~\ref{def:dualcp}), we see that $\mathcal{S}(\mathcal{F})$ is a composition of three channels, i.e. CP trace-preserving maps. 

The first of these channels has the following dilation:
\begin{align}\label{eq:circuitdil1}
\includegraphics[valign=c,scale=1]{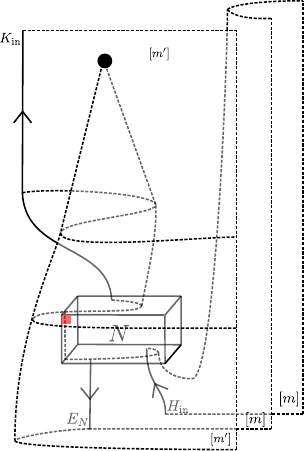}
\end{align}
Looking at the type of the dilation, we see this is a channel 
$$
\bigoplus_{k \in [m']} B(\Kin{}_{,k}) \to \bigoplus_{\substack{k' \in [m'] \\ i \in [m] \\ i' \in [m]}} B(E_{N,ik}) \otimes B(\Hin{}_{,i'}).
$$
In order to recover the circuit of Theorem~\ref{thm:mainintro}, we need to get rid of the $i,k$-dependence of the family of Hilbert spaces $\{E_{N,ik}\}$. For this, let $P$ be the Hilbert space from the set $\{E_{N,ik}\}_{i,k \in [m] \times [m']}$ of highest dimension, and pick isometries $\iota_{ik}: E_{N,ik} \hookrightarrow P$ for each $i,k$. This specifies an isometric 2-morphism $\iota: E_N \to \widetilde{E}_{N}$, where $\widetilde{E}_{N,ik}:= P$. Composing~\eqref{eq:circuitdil1} with the coisometry $\iota^{\dagger}$, we obtain a channel 
\begin{equation}\label{eq:circuitchan1}
\bigoplus_{k \in [m']} B(\Kin{}_{,i}) \to \bigoplus_{\substack{k' \in [m'] \\ i \in [m] \\ i' \in [m]}} B(P) \otimes B(\Hin{}_{,i'}).
\end{equation}
To take account of this added coisometry we will modify the last of the three composed channels later on; we will do this in such a way that the total channel $\mathcal{S}(\mathcal{F})$ stays the same. 

Looking at the connectivity of the regions in~\eqref{eq:circuitdil1}, we see that the index $k \in [m']$ in the target is constrained to be the same as the index $k' \in [m']$ in the source; moreover, the index $i \in [m]$ is constrained to be the same as the index $i' \in [m]$. We thus obtain a channel
$$
\mathcal{E}: \bigoplus_{k \in [m']} B(\Kin{}_{,i}) \to \bigoplus_{\substack{i \in [m]}} B(P) \otimes B(\Hin{}_{,i}).
$$
from which the channel~\eqref{eq:circuitchan1} can be recovered by copying of classical information beforehand and afterwards, as in the circuit~\eqref{eq:realisationcircuit}.

The second of the channels defining $\mathcal{S}(\mathcal{F})$ has the following dilation:
\begin{align*}
    \includegraphics[valign=c,scale=1]{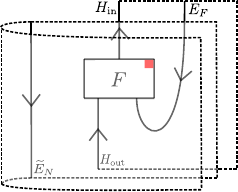}
\end{align*}
This is the dilation of the channel 
$$
\mathcal{F}: \bigoplus_{i \in [m]} B(\Hin{}_{,i}) \to \bigoplus_{j \in [n]} B(\Hout{}_{,j}), 
$$
as depicted in the circuit diagram~\eqref{eq:realisationcircuit}.

The third and final channel defining $\mathcal{S}(\mathcal{F})$ has the following dilation:
\begin{align}\label{eq:circuitchan3}
\includegraphics[valign=c,scale=1]{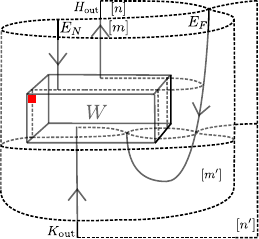}
\end{align}
We need to deal with the isometry $\iota$ we added earlier, because the dilation~\eqref{eq:circuitchan3} has input wire $E_N$ rather than $\tilde{E}_N$. First (recalling the definition of the direct sum of 1-morphisms~\eqref{eq:directsum1morphisms}) we define a 1-morphism $Z: [n] \boxtimes [m] \to [n'] \boxtimes [m']$, as follows:
\begin{align*}
\includegraphics[valign=c,scale=1]{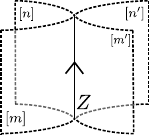}~~:=~~
\includegraphics[valign=c,scale=1]{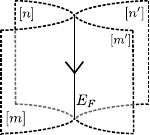}
~~\bigoplus~~
\includegraphics[valign=c,scale=1]{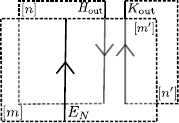}
\end{align*}
By definition of the direct sum, we have isometries
\begin{align*}
\includegraphics[valign=c,scale=1]{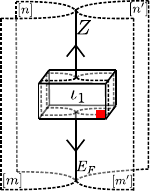}
&&
\includegraphics[valign=c,scale=1]{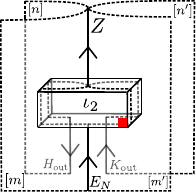}
\end{align*}
Let $E_{N,ik}^{\bot} \subset P$ be the orthogonal complement of $\iota_{ik}(E_{N,ik}) \subset P$; this defines a 1-morphism $E_N^{\bot}: [m'] \to [m]$, and we can now (recalling the definition of the direct sum of 2-morphisms~\eqref{eq:directsum2morphisms}) define the following dilation 
\begin{align}\label{eq:finaldil}
\includegraphics[valign=c,scale=1]{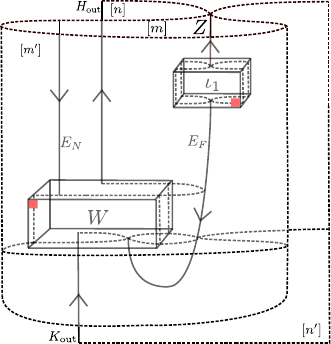}
~~ \bigoplus ~~
\includegraphics[valign=c,scale=1]{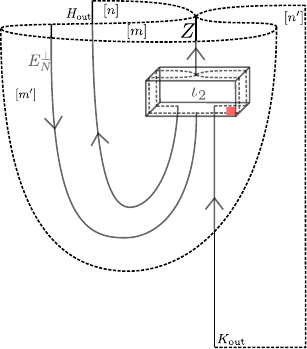}
\end{align}
of a completely positive map whose type (by distributivity of the direct sum of Hilbert spaces over the tensor product) is
$$
\bigoplus_{\substack{k \in [m'] \\ j \in [n] \\ i \in [m]}} B(P) \otimes B(\Hout{}_{,j}) \to \bigoplus_{l \in [n']}B(\Kout{}_{,l}).
$$
This is precisely the type of the last instrument in the circuit diagram~\eqref{eq:realisationcircuit}. It can readily be checked that this CP map is trace preserving; moreover, the overall channel $\mathcal{S}(\mathcal{F})$ is preserved by this modification of the last channel together with the inclusion of the coisometry $\iota^{\dagger}$, since the second factor of~\eqref{eq:finaldil} is annihilated by $\iota^{\dagger}$.

To finish the proof we need to derive the dimensional limit on the Hilbert space $P$. For this we need simply note the minimality condition, namely invertibility of  the 2-morphism~\eqref{eq:minimalitycond}: since $(N,E_N)$ is a minimal Stinespring representation, this implies that the following 2-morphism has a left inverse:
\begin{align*}
\includegraphics[valign=c,scale=1]{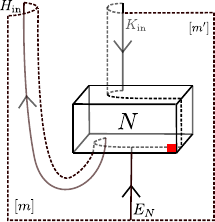}
\end{align*}
Considering this 2-morphism as an indexed family of linear maps, this implies that 
$$
\dim(E_{N,ik}) \leq \dim(\Hin{}_{,i})\dim(\Kin{}_{,k})
$$
for all $i,k$, which, given $\dim(P) = \max_{i,k}\left(\dim(E_{N,ik})\right)$, yields the result. This completes the proof of Theorem~\ref{thm:mainintro}. 
\qed

\end{document}